\documentclass[copyright,creativecommons]{eptcs}
\providecommand{\event}{A Second Soul: Celebrating the Many Languages of Programming \\ Festschrift in Honor of Peter Thiemann's Sixtieth Birthday}

\usepackage{iftex}

\ifpdf
  \usepackage{underscore}         %
  \usepackage[T1]{fontenc}        %
\else
  \usepackage{breakurl}           %
\fi

\usepackage{amsmath,amsfonts,amssymb,amsthm,cite,enumitem}
\usepackage{graphicx}

\theoremstyle{definition}
\newtheorem{definition}{Definition}
\newtheorem{property}{Property}
\newtheorem{theorem}{Theorem}
\newtheorem{lemma}{Lemma}
\newtheorem{corollary}{Corollary}

\newcommand{\lemmaref}[1]{Lemma~\ref{#1}}
\newcommand{\corollaryref}[1]{Corollary~\ref{#1}}

\newcommand{\appendixref}[1]{Appendix~\ref{#1}}
\newcommand{\sectionref}[1]{Section~\ref{#1}}
\newcommand{\sectionsref}[2]{Sections~\ref{#1} and~\ref{#2}}

\newcommand{\eg}{{\rm e.g.}}
\newcommand{\ie}{{\rm i.e.}}

\newcommand{\elam}[2]{\ensuremath{\lambda{#1}.{#2}}}
\newcommand{\elamp}[2]{\ensuremath{(\elam{#1}{#2})}}
\newcommand{\fapp}[2]{\ensuremath{{#1}({#2})}}
\newcommand{\eapp}[2]{\ensuremath{{#1}\,{#2}}}
\newcommand{\eappp}[2]{\ensuremath{(\eapp{#1}{#2})}}
\newcommand{\eappq}[2]{\ensuremath{\eapp{#1}{(#2)}}}
\newcommand{\expo}[2]{\ensuremath{{#1}^{#2}}}
\newcommand{\equofloor}[2]{\ensuremath{\lfloor \frac{#1}{#2} \rfloor}}

\newcommand{\estreamthree}[4]{\ensuremath{[{#1},\,{#2},\,{#3},\cdots[}}
\newcommand{\estreamthreea}[4]{\ensuremath{[{#1},}$ ${{#2},\,{#3},\cdots[}}
\newcommand{\estreamthreeb}[4]{\ensuremath{[{#1},\,{#2},}$ ${{#3},\cdots[}}
\newcommand{\estreamthreec}[4]{\ensuremath{[{#1},\,{#2},\,{#3},}$ ${\cdots[}}
\newcommand{\estreamfour}[4]{\ensuremath{[{#1},\,{#2},\,{#3},\,{#4},\cdots[}}
\newcommand{\estreamfourc}[4]{\ensuremath{[{#1},\,{#2},\,{#3},}$ ${{#4},\cdots[}}

\newcommand{\estreamfourd}[4]{\ensuremath{[{#1},\,{#2},\,{#3},\,{#4},}$ ${\cdots[}}
\newcommand{\estreamfive}[5]{\ensuremath{[{#1},\,{#2},\,{#3},\,{#4},\,{#5},\cdots[}}
\newcommand{\estreamsix}[6]{\ensuremath{[{#1},\,{#2},\,{#3},\,{#4},\,{#5},\,{#6},\cdots[}}
\newcommand{\estreamseven}[7]{\ensuremath{[{#1},\,{#2},\,{#3},\,{#4},\,{#5},\,{#6},\,{#7},\cdots[}}

\newcommand{\estreameight}[8]{\ensuremath{[{#1},\,{#2},\,{#3},\,{#4},\,{#5},\,{#6},\,{#7},\,{#8},\cdots[}}
\newcommand{\estreamnine}[9]{\ensuremath{[{#1},\,{#2},\,{#3},\,{#4},\,{#5},\,{#6},\,{#7},\,{#8},\,{#9},\cdots[}}
\newcommand{\rawantidifference}{\ensuremath{\nabla^{-1}}}
\newcommand{\eantidifferenceone}[1]{\ensuremath{\rawantidifference\,{#1}}}
\newcommand{\eantidifferenceonep}[1]{\ensuremath{(\eantidifferenceone{#1})}}
\newcommand{\rawdifference}{\ensuremath{\nabla}}
\newcommand{\edifferenceone}[1]{\ensuremath{\rawdifference\,{#1}}}

\newcommand{\edifference}[2]{\ensuremath{\rawdifference\,{#1}\,{#2}}}

\newcommand{\econditional}[3]{\ensuremath{\mathrm{if}\;{#1}\;\mathrm{then}\;{#2}\;\mathrm{else}\;{#3}}}
\newcommand{\emod}[2]{\ensuremath{{#1}\;\mathrm{mod}\;{#2}}}

\newcommand{\epair}[2]{\ensuremath{({#1},\,{#2})}}

\newcommand{\rawgee}{\ensuremath{g}}
\newcommand{\gee}[2]{\eapp{\eapp{\rawgee}{#1}}{#2}}
\newcommand{\rawgeeop}{\ensuremath{\overline{\rawgee}}}
\newcommand{\geeop}[2]{\eapp{\eapp{\rawgeeop}{#1}}{#2}}

\newcommand{\rawgeen}[1]{\ensuremath{g}}
\newcommand{\geen}[3]{\eapp{\eapp{\rawgeen{#1}}{#2}}{#3}}
\newcommand{\rawgeenop}[1]{\ensuremath{\overline{\rawgeen{#1}}}}

\newcommand{\mulset}[2]{\ensuremath{\left(\!\binom{#1}{#2}\!\right)}}
\newcommand{\mmulset}[2]{\ensuremath{\left(\!\!\binom{#1}{#2}\!\!\right)}}

\usepackage{color}
\definecolor{ltblue}{rgb}{0,0.0,0.0}
\definecolor{dkblue}{rgb}{0,0.0,0.0}
\definecolor{dkgreen}{rgb}{0,0.0,0}
\definecolor{dkviolet}{rgb}{0.0,0,0.0}
\definecolor{dkred}{rgb}{0.0,0,0}
\usepackage{listings}
\newcommand{\inputscheme}[1]{\lstinputlisting[linerange={#1}-END]{scheme/streamless.scm}}
\newcommand{\inlinescheme}[1]{\lstinline[basicstyle=\small\ttfamily]{#1}}

\def\doframeit#1{\vbox{%
  \hrule height\fboxrule
    \hbox{%
      \vrule width\fboxrule \kern\fboxsep
      \vbox{\kern\fboxvsep #1\kern\fboxvsep }%
      \kern\fboxsep \vrule width\fboxrule }%
    \hrule height\fboxrule }}

\def\frameit{\smallskip \advance \linewidth by -7.5pt \setbox0=\vbox \bgroup
\strut \ignorespaces }

\def\endframeit{\ifhmode \par \nointerlineskip \fi \egroup
\doframeit{\box0}}

\newdimen \fboxvsep
\fboxvsep=\fboxsep
\advance \fboxsep by -7.5pt

\title{
  Summa Summarum: \\
  Moessner's Theorem without Dynamic Programming
}
\author{Olivier Danvy
  \institute{School of Computing \& Yale-NUS College, National University of Singapore, Singapore}
  \email{danvy@acm.org}
}
\def\titlerunning{Moessner's Theorem without Dynamic Programming}
\def\authorrunning{O.~Danvy}
\begin{document}
\maketitle

\begin{abstract}
\noindent
Seventy years on, Moessner's theorem and Moessner's process -- \ie, the
additive
computation
of
integral
powers -- continue to fascinate.
They have given rise
to a variety of elegant proofs,
to an implementation in hardware,
to
generalizations,
and now even to a popular video, ``The Moessner Miracle.''
The existence of this video, and even more its title, indicate that while
the ``what'' of Moessner's process is understood, its ``how'' and even more
its ``why'' are still elusive.
And indeed all the proofs of Moessner's theorem
involve more complicated concepts than both the theorem and the process.

This
article
identifies that Moessner's process implements an additive
function with dynamic programming.
A version of this implementation without dynamic programming
(1) gives rise to a simpler statement of Moessner's theorem and
(2) can be abstracted and then instantiated
into
related additive
computations.
The simpler statement also suggests a simpler and more efficient
implementation
to compute
integral powers as well as simple additive functions to compute, \eg, Factorial
numbers.
It also reveals the source of -- to quote John Conway and Richard Guy -- Moessner's magic.
\end{abstract}

\noindent
Keywords:
Moessner's theorem;
Moessner's process;
streams;
nested summations;
primitive iteration;
primitive recursion;
integral powers;
single, double, triple, etc.~Factorial numbers;
superfactorial numbers;
binomial coefficients;
Catalan numbers;
Fibonacci numbers;
Euler (zigzag, up/down) numbers;
polygonal numbers (a new characterization of)

\

\hrule

\ 

\begin{center}
\begin{Large}
Executive Summary
\end{Large}
\end{center}

\noindent
Beneath its dynamic-programming infrastructure,
the foundation of Moessner's process is nested summations.
Here is Moessner's theorem without dynamic programming,
where each nested sum is a streamless instance of a prefix sum and each inner upper bound
is a streamless instance of the filtering-out phase in Moessner's process:
\begin{eqnarray*}
  \forall x : \mathbb{N},\forall n : \mathbb{N},
  \sum_{i_1=0}^x\:
  \sum_{i_2=0}^{\equofloor{2 \cdot i_1}{1}}\:
  \sum_{i_3=0}^{\equofloor{3 \cdot i_2}{2}}\:
  \cdots
  \sum_{i_n=0}^{\equofloor{n \cdot i_{n-1}}{n-1}}\:
  1
  \;\;=\;\;
  \sum_{i_1=0}^x\,
  \sum_{i_2=0}^{i_1 + \equofloor{i_1}{1}}\:
  \sum_{i_3=0}^{i_2 + \equofloor{i_2}{2}}\:
  \cdots
  \sum_{i_n=0}^{i_{n-1} + \equofloor{i_{n-1}}{n-1}}\:
  1
  & = &
  (x + 1)^n
\end{eqnarray*}

\noindent
Moessner's magic~\cite{Conway-Guy:TBON96-Moessner-s-Magic} lies in the dependencies of the
indices in these
summations.

Here are some new corollaries: $\forall x : \mathbb{N},\forall n : \mathbb{N},$
\\[-3mm]
$$
 \begin{array}{@{}r@{\ \ }l@{\ }l@{\ }l@{\ }l@{\ }l@{\ }l@{\ \ }l@{\ \ }l@{}}
   \textrm{Binomial coefficients:}
   &
   \sum_{i_1=0}^x\,
   &
   \sum_{i_2=0}^{i_1}\,
   &
   \sum_{i_3=0}^{i_2}\,
   &
   \cdots
   &
   \sum_{i_n=0}^{i_{n-1}}\,
   &
   1
   &
   =
   &
   \binom{x + n}{n}
   \\[3mm]
   \textrm{Catalan numbers:}
   &
   \sum_{i_1=0}^0\,
   &
   \sum_{i_2=0}^{i_1 + 1}\,
   &
   \sum_{i_3=0}^{i_2 + 1}\,
   &
   \cdots
   &
   \sum_{i_n=0}^{i_{n-1} + 1}\,
   &
   1
   &
   =
   &
   \frac{\binom{2 \cdot n}{n}}{n + 1}
   \:\, = \;\,
   C_n
   \\[4mm]
   \textrm{Fibonacci numbers:}
   &
   \sum_{i_1=0}^0\,
   &
   \sum_{i_2=0}^{1 - i_1}\,
   &
   \sum_{i_3=0}^{1 - i_2}\,
   &
   \cdots
   &
   \sum_{i_n=0}^{1 - i_{n-1}}\,
   &
   1
   &
   =
   &
   F_{n+1}
   \\[4mm]
   \textrm{Euler (zigzag) numbers:}
   &
   \sum_{i_1=0}^0\,
   &
   \sum_{i_2=0}^{n - 1 - i_1}\,
   &
   \cdots
   &
   \sum_{i_{n-1}=0}^{1 - i_{n-2}}\,
   &
   \sum_{i_n=0}^{0 - i_{n-1}}\,
   &
   1
   &
   =
   &
   E_n
 \end{array}$$

\noindent
Generalizing $0$ to $x$ in the second corollary
gives rise to
the $x + 1$st convolution of Catalan numbers.

\clearpage

That said, for several of the applications of Moessner's theorem in the literature
-- \eg, factorial numbers but also integral powers --
Moessner's magic (\ie, the upper bound of an inner sum depending on the index of an outer sum)
is not needed.
In each of the following summations, the index does not occur in the body of this summation:
\\[-3mm]
$$
  \begin{array}{@{}r@{\ \ }r@{\ }l@{\ }l@{\ }l@{\ }l@{\ }l@{\ }l@{\ \ }l@{\ \ }l@{}}
   \textrm{Factorial numbers:}
   &
   \forall n : \mathbb{N},
   &
   \sum_{i_1=0}^1\,
   &
   \sum_{i_2=0}^2\,
   &
   \sum_{i_3=0}^3\,
   &
   \cdots
   &
   \sum_{i_n=0}^n\,
   &
   1
   &
   =
   &
   (n + 1)!
   \\[3mm]
   \textrm{Integral powers:}
   &
   \forall x, \, n : \mathbb{N},
   &
   \sum_{i_1=0}^x\,
   &
   \sum_{i_2=0}^x\,
   &
   \sum_{i_3=0}^x\,
   &
   \cdots
   &
   \sum_{i_n=0}^x\,
   &
   1
   &
   =
   &
   (x + 1)^n
 \end{array}
$$

This observation sheds a new light on Paasche's slide rule~\cite{Paasche:MNU53-54}
since the product of a sequence of factors can be expressed as an iterated summation:
$$
 \begin{array}{@{}r@{\ }c@{\ \ }c@{\ \ }l@{\ }l@{\ }l@{\ }l@{\ }l@{\ }l@{}}
   \forall f : \mathbb{N} \rightarrow \mathbb{N},
   \forall n : \mathbb{N},
   &
   \prod_{i=0}^n\,\fapp{f}{i}
   &
   =
   &
   \sum_{i=1}^{\fapp{f}{0}}\,
   &
   \sum_{i=1}^{\fapp{f}{1}}\,
   &
   \sum_{i=1}^{\fapp{f}{2}}\,
   &
   \cdots
   &
   \sum_{i=1}^{\fapp{f}{n}}\,
   &
   1
 \end{array}$$

In practice, the nested summations induce many duplicated computations.
Dynamic programming makes it possible to avoid these duplications,
giving rise, \eg, to Pascal's triangle for binomial coefficients
and more generally to additive processes in the style of Moessner
that are ready to be implemented in hardware~\cite{Samadi-al:ISCAS05}.

\ 

\hrule

\section{Introduction}
\label{sec:introduction}

Whereas Eratosthenes's sieve constructs a stream of successive prime numbers $\estreamthree{2}{3}{5}{7}$
given the stream of positive natural numbers $\estreamthree{1}{2}{3}{4}$,
Moessner's
process
constructs a stream of integral powers $\estreamthree{1^n}{2^n}{3^n}{4^n}$
given
the stream $\estreamthree{1}{0}{0}{0}$
and
a natural number $n$.
Counting down from $n$ to $0$,
\begin{description}[leftmargin=3.5mm,itemsep=1pt]

\item{($n$)}
it
filters out (strikes out / skips / drops / elides / omits)
each $n + 2$nd element from this given stream
and
then
constructs the corresponding stream of prefix sums
(given a stream
$\estreamthree{x_0}{x_1}{x_2}{x_3}$,
the corresponding stream of prefix sums is 
$\estreamthreeb{\sum_{i=0}^0\,x_i}{\sum_{i=0}^1\,x_i}{\sum_{i=0}^2\,x_i}{\sum_{i=0}^3\,x_i}$,
\ie,
$\estreamthree{x_0}{x_0 + x_1}{x_0 + x_1 + x_2}{x_0 + x_1 + x_2 + x_3}$);

\item{}
\ldots

\item{($i$)}
it
filters out each $i + 2$nd element from the resulting stream
and
then
constructs the corresponding stream of prefix sums;

\item{}
\ldots

\item{($0$)}
it
filters out each second element from the resulting stream
and
then
constructs the corresponding stream of prefix sums.

\end{description}

Moessner's theorem~\cite{Moessner:51} states that 
the resulting stream is $\estreamthree{1^n}{2^n}{3^n}{4^n}$,
which is immediately visible
for small values of $n$:
\begin{itemize}[leftmargin=3.5mm,itemsep=1pt]

\item
  When $n = 0$, the resulting stream of prefix sums is $\estreamfourd{1}{1 + 0}{(1 + 0) + 0}{((1 + 0) + 0) + 0}$,
  \ie,
  $\estreamfour{1}{1}{1}{1}$,
  or again $\estreamfour{1^0}{2^0}{3^0}{4^0}$.

\item
  When $n = 1$, the first stream of prefix sums is the stream of $1$'s
  and the resulting stream of prefix sums is
  $\estreamfour{1}{1 + 1}{(1 + 1) + 1}{((1 + 1) + 1) + 1}$,
  \ie,
  $\estreamfourc{1}{2}{3}{4}$,
  or again $\estreamfour{1^1}{2^1}{3^1}{4^1}$.

\item
  When $n = 2$,
  the first stream of prefix sums is the stream of $1$'s and
  the second is the stream of positive natural numbers.
  Filtering out its second elements constructs the stream of odd natural numbers,
  and the resulting stream of prefix sums is
  $\estreamfour{1}{1 + 3}{(1 + 3) + 5}{((1 + 3) + 5) + 7}$,
  \ie,
  $\estreamfive{1}{4}{9}{16}{25}$,
  or again $\estreamfive{1^2}{2^2}{3^2}{4^2}{5^2}$.
  Indeed, as is known since the Pythagoreans, the $x + 1$st positive square integer
  is the sum of the first $x + 1$st odd natural numbers:
\label{page:sum-of-first-odd-natural-numbers}
  $$
    \forall x : \mathbb{N},
    \sum_{i=0}^x\,(2 \cdot i + 1)
    =
    (x + 1)^2
  $$
  since $2 \cdot i + 1 = (i + 1)^2 - i^2$ because of
  the binomial expansion into a sum of
  monomials with binomial coefficients
  $(i + 1)^2 = \binom{2}{2} \cdot i^2 \cdot 1^0 + \binom{2}{1} \cdot i^1 \cdot 1^1 + \binom{2}{0} \cdot i^0 \cdot 1^2 = i^2 + 2 \cdot i + 1$
  and using
  the telescoping sum
  $\sum_{i=0}^x\,(i + 1)^2 - i^2
   = (1^2 - 0^2) + (2^2 - 1^2) + (3^2 - 2^2) + \cdots + (x^2 - (x - 1)^2) + ((x + 1)^2 - x^2)
   = -0^2 + (x + 1)^2$.

\label{page:telescoping-sum}

\end{itemize}

\noindent
Beyond $2$, things become more involved, but the result is still
remarkable because not a single multiplication is performed to obtain powers.
Moessner's process has therefore been implemented in hardware for
signal-processing purposes~\cite{Samadi-al:ISCAS05}.

Over the years, a number of proofs for Moessner's theorem have been put forward:
first, induction proofs by mathematicians in the 20th century
(Long, Paasche, Perron, Sali\'e, Slater, van Yzeren,
and then Ross,
as reviewed in \sectionref{subsec:in-the-20th-century})
and then, both induction proofs and coinduction proofs by theoretical computer scientists in the 21st century
(Bickford, Hinze, Kozen, Krebbers, Niqui, Parlant, Rutten, and Silva,
as reviewed in \sectionref{subsec:in-the-21st-century}).
Nowadays, Moessner's process is a showcase for
demonstrating stream calculi, mechanizing coinduction proofs, discovering patterns, and formulating generalizations.
(Watch the video~\cite{Polster:21}.)

The starting point here is the observation
that Moessner's iterative process of computing successive streams
fits the pattern of dynamic programming.
The goal of this
article
is to unveil the essence of Moessner's theorem by removing this dynamic-programming infrastructure.

\subsection{Roadmap}

\sectionref{sec:background-and-related-work}
illustrates
Moessner's process and reviews
previous work about Moessner's theorem.
\sectionref{sec:a-polynomial-rendition-of-moessner-s-process}
presents a version of Moessner's process 
that does not use dynamic programming.
\sectionref{sec:a-polynomial-rendition-of-the-left-inverse-of-moessner-s-process}
presents a version of the left inverse of Moessner's process 
that also does not use dynamic programming.
\sectionref{sec:the-essence-of-moessner-s-theorem}
builds on \sectionsref{sec:a-polynomial-rendition-of-moessner-s-process}{sec:a-polynomial-rendition-of-the-left-inverse-of-moessner-s-process}
and presents the essence of Moessner's theorem.
Correspondingly, \sectionref{sec:the-essence-of-moessner-s-process}
implements the essence of Moessner's process.
\sectionref{sec:a-parameterized-implementation-of-moessner-s-process}
parameterizes this implementation and illustrates it with several instantiations
that reveal Moessner's magic as
primitive recursion.
\sectionref{sec:back-to-dynamic-programming} goes back to dynamic programming,
using lists instead of streams.
\sectionref{sec:related-work-about-nested-sums} reviews previous work about nested sums.
\appendixref{app:a-curiosa-about-polygonal-numbers} presents a new characterization of polygonal numbers
in memory of Moessner.

\subsection{Prerequisites and Notations}
\label{subsec:prerequisites-and-notations}

An elementary grasp of functional programming
(\eg, the $\lambda$ notation for functions~\cite{Church:41}
and the \inlinescheme{lambda} notation for procedures in Scheme~\cite{Dybvig:96-not-online})
and of discrete mathematics~\cite{Graham-Knuth-Patashnik:94,Rosen:19}
is expected from the reader.
Also, we use Peano numbers, \ie, natural numbers that start with 0.

\subsubsection{Functional programming}

The programming language of discourse here is Scheme,
a block-struc\-tured and lexically scoped dialect of Lisp with
first-class procedures and proper tail recursion.
(So programs are fully parenthesized expressions that use the Polish prefix notation
and iteration is achieved with tail-recursive procedures.)

A list is an inductive data structure that is constructed with nil (the
empty list, written \inlinescheme{'()} in Scheme) in the base case, and
with cons (implemented by the Scheme procedure \inlinescheme{cons}) in the induction step.
For example, the Scheme procedure \inlinescheme{iota}, given a non-negative integer $n$,
constructs the list of $n$ payloads 0, 1, 2, \ldots, $(n-1)$.
So evaluating \inlinescheme{(iota 3)} gives rise to evaluating
\inlinescheme{(cons 0 (cons 1 (cons 2 '())))} and yields the list \inlinescheme{(0 1 2)}.

Likewise, the Scheme procedure \inlinescheme{map}, given a procedure and a list,
applies this procedure to each payload in the list and constructs the list of the results,
in the same order.
So for example, evaluating
\begin{center}
\inlinescheme{(map (lambda (n) (+ n 1)) (iota 3))}
\end{center}
yields \inlinescheme{(1 2 3)} since \inlinescheme{(lambda (n) (+ n 1))}
implements the successor function, which is applied to each of the
payloads in \inlinescheme{(0 1 2)} here.

A stream is a coinductive data structure (intuitively: an unbounded list~\cite{Friedman-Wise:ICALP76}) that is constructed on demand.
Given a function $f : \mathbb{N} \rightarrow \mathbb{N}$, the stream
$\estreamthree{\eapp{f}{0}}{\eapp{f}{1}}{\eapp{f}{2}}{\eapp{f}{3}}$
is a lazy representation of $f$'s function graph, where each pair
$(i,\:\eapp{f}{i})$ is represented with $i$ and with the payload of the stream at index $i$.
The stream $\estreamthree{\eapp{f}{0}}{\eapp{f}{1}}{\eapp{f}{2}}{\eapp{f}{3}}$
is said to \emph{enumerate} the function $f$.

For example, in Scheme, one can
implement
a procedure \inlinescheme{lazy-iota} that,
when applied to an integer $n$, constructs the stream of payloads $\estreamthreeb{n}{n+1}{n+2}{n+3}$.
One can also implement a procedure \inlinescheme{lazy-map} that, given a procedure and a stream,
applies this procedure to each payload in the stream and constructs the stream of the results,
in the same order.
Given a procedure \inlinescheme{f} that implements the function $f : \mathbb{N} \rightarrow \mathbb{N}$, evaluating
\vspace{-1mm}
\begin{center}
\inlinescheme{(lazy-map f (lazy-iota 0))}
\end{center}
\vspace{-1mm}
yields a stream that enumerates $f$.

\subsubsection{Differences and antidifferences}
\label{subsubsec:differences-and-antidifferences}

The prefix sum (a.k.a.~antidifference) of a function $f : \mathbb{N} \rightarrow \mathbb{N}$ is
$\eantidifferenceone{f}$:
$$
\begin{array}{l@{\ }c@{\ }l}
\eantidifferenceone{f}
& = & 
\elam{n}{\sum_{i=0}^n\,\eapp{f}{i}}
\end{array}
$$
It is used in Moessner's process (but we shall not use the
$\rawantidifference$ notation).

The backward difference of a function $f : \mathbb{N} \rightarrow \mathbb{N}$ is $\edifferenceone{f}$:
$$
\edifference{f}{0} = \eapp{f}{0}
\;\;\; \wedge \;\;\;
\forall n : \mathbb{N},
\edifference{f}{(n+1)} = \eapp{f}{(n+1)} - \eapp{f}{n}
$$
It is used in the inverse of Moessner's process (\sectionref{subsec:inverting-moessner-s-process-streamlessly}).

$\rawantidifference$ is a right inverse (as well as a left inverse) of $\rawdifference$:
$$
\left\{
\begin{array}{l@{\ }c@{\ }l}
\edifference{\eantidifferenceonep{f}}{0}
& = &
\edifference{\elamp{n}{\sum_{i=0}^n\,\eapp{f}{i}}}{0}
\\[1mm]
& = &
\sum_{i=0}^0\,\eapp{f}{i}
\\[1mm]
& = & 
\eapp{f}{0}
\\[3mm]
\edifference{\eantidifferenceonep{f}}{(n+1)}
& = &
\edifference{\elamp{n}{\sum_{i=0}^n\,\eapp{f}{i}}}{(n+1)}
\\[1mm]
& = &
\eapp{\elamp{n}{\sum_{i=0}^n\,\eapp{f}{i}}}{(n+1)} - \eapp{\elamp{n}{\sum_{i=0}^n\,\eapp{f}{i}}}{n}
\\[1mm]
& = & 
\eapp{f}{(n+1)}
\end{array}
\right.
$$

\noindent
Also, for all expressions $e : \mathbb{N}$,
the equality $\sum_{i=0}^x\,e = (x + 1) \cdot e$ holds
when the local variable $i$ does not occur free in the expression $e$.

\subsubsection{Dynamic programming}

Dynamic programming~\cite{Cormen-al:01} is a way of writing programs
so that the results of overlapping local computations are memoized
and then shared.
This sharing ensures that the overlapping local computations are
not repeated in the course of a global computation.

\begin{itemize}[leftmargin=3.5mm]

\item
  Consider Fibonacci numbers, for example.
  They are inductively specified with the following second-order linear recurrence:
  $$
    F_0 = 0 \;\;\; \wedge \;\;\; F_1 = 1 \;\;\; \wedge \;\;\; \forall n : \mathbb{N}, F_{n+2} = F_n + F_{n+1}
  $$
  A recursive function mapping $n : \mathbb{N}$ to $F_n$
  can be implemented based on this inductive specification.
  As is well known,
  this additive implementation gives rise to overlapping intermediate computations
  in the form of many repeated identical recursive calls.

  In contrast, using dynamic programming, one can implement an additive
  program that iterates over successive pairs of Fibonacci numbers,
  starting from $\epair{F_0}{F_1}$ and without repeating any computation.
  Given a number $n : \mathbb{N}$,
  the result $F_n$ is found in one of the components of the resulting pair.

  For example, it is known since Burstall and Darlington~\cite{Burstall-Darlington:JACM77}
  how to map $\epair{F_0}{F_1}$ to $\epair{F_n}{F_{n+1}}$ in $n$ tail-recursive calls
  to a local procedure \inlinescheme{visit}:
  \inputscheme{FIB}

  The Scheme procedure \inlinescheme{fib} is applied to an integer $n$,
  verifies that this integer is non-negative,
  and iterates over successive triples of integers -- a decreasing counter and two increasing Fibonacci numbers:
  it starts with $n$, $0$, and $1$
  and continues with $n-1$, $1$, and $1$
  and then with $n-2$, $1$, and $2$
  and then with $n-3$, $2$, and $3$,
  until \inlinescheme{visit} is applied to $0$, $F_n$, and $F_{n+1}$.
  At that point, the iteration stops and $F_n$ is returned.

  Dynamic programming is at work here because the structure of the overlapping computations
  makes it possible to compute each intermediate Fibonacci number only once.
  This structure can be exploited in many ways.
  For example, programs written using dynamic programming need not only be additive:
  a dynamic program that uses multiplications and divisions
  can compute Fibonacci numbers with a logarithmic rather than linear complexity.

\item
  For another example, consider binomial coefficients.
  They are inductively specified as follows (the third clause is Pascal's
  rule): $\forall n : \mathbb{N},$
  $$
  \binom{n}{0} = 1 \;\;\; \wedge \;\;\;
  \binom{n}{n} = 1 \;\;\; \wedge \;\;\;
  \forall k : \mathbb{N},\:k < n, \binom{n+1}{k+1} = \binom{n}{k} + \binom{n}{k+1}
  $$
  A recursive function mapping $n : \mathbb{N}$ and $k : \mathbb{N}$,
  where $k \le n$, to $\binom{n}{k}$
  can be implemented based on this inductive specification.
  As is well known,
  this additive implementation gives rise to overlapping intermediate computations
  in the form of many repeated identical recursive calls.

  In contrast, using dynamic programming, one can implement an additive
  program that iterates over successive lists of binomial coefficients,
  starting from $[\binom{0}{0}]$ and without repeating any computation.
  Given a number $n : \mathbb{N}$,
  the result $\binom{n}{k}$, for any $k : \mathbb{N}$ such that $k \le n$, 
  is found at index $k$ in the resulting list
  $[\binom{n}{0},\,\binom{n}{1},\cdots\!,\,\binom{n}{n-1},\,\binom{n}{n}]$.
  These successive lists form Pascal's triangle, an early example of
  dynamic programming.

\end{itemize}

The art of dynamic programming is to start from the recursive function
that implements an inductive specification (\eg, of Fibonacci numbers
or of binomial coefficients), to identify the overlapping computations,
and to devise
a process 
to compute intermediate results only once to obtain the desired result.
The present article is about the converse: start from the computational process
and reverse-engineer a recursive function.
The starting point here is that Moessner's
process of constructing successive streams fits the pattern of dynamic
programming, and the point of
\sectionref{sec:a-polynomial-rendition-of-moessner-s-process} is to
reverse-engineer a recursive function that
implements
an inductive specification of integral powers.
The thesis defended here is that 
the essence of Moessner's theorem is in this inductive specification,
not in the dynamic-program\-ming infrastructure of Moessner's process.

\section{Background and Related Work}
\label{sec:background-and-related-work}

Originally~\cite{Moessner:51}, Moessner's process started with $\estreamthree{1}{2}{3}{4}$.
Paasche~\cite{Paasche:52} pointed out that it could start with $\estreamthree{1}{1}{1}{1}$,
which has the effect of undoing one iteration of the process.
Later on~\cite{Clausen-al:TCS14}, the author observed that one more iteration could be undone
and that the process could start with $\estreamthree{1}{0}{0}{0}$.
This observation is used in \sectionref{sec:introduction} for presentational purposes
and revisited in \sectionref{subsec:long-s-second-theorem} to recast Long's second theorem~\cite{Long:AMM66}.
In the rest of this
article,
the process starts with $\estreamthree{1}{1}{1}{1}$,
as per Paasche's suggestion,
which shortcuts the initial iteration from $\estreamthree{1}{0}{0}{0}$ to $\estreamthree{1}{1}{1}{1}$.

Let us
illustrate Moessner's process for the exponents 0, 1, 2, and 3.
\begin{description}[leftmargin=0mm]

\item[Exponent 0:] 
No
iterations take place. %
\begin{description}[leftmargin=3.5mm]

\item{}

The result is the starting stream $\estreamthree{1}{1}{1}{1}$,
\ie, $\estreamthree{1^0}{2^0}{3^0}{4^0}$,
the stream of the
positive integers exponentiated with $0$.

\end{description}

\item[Exponent 1:] One iteration takes place. %
\begin{description}[leftmargin=3.5mm]

\item{(0)}
Filtering out each second element of
$\estreamthree{1}{1}{1}{1}$
vacuously constructs $\estreamthree{1}{1}{1}{1}$.
The stream of the resulting prefix sums is $\estreamthree{1}{2}{3}{4}$,
\ie, $\estreamthree{1^1}{2^1}{3^1}{4^1}$,
the stream of the
positive integers exponentiated with $1$.

\end{description}

\item[Exponent 2:] Two iterations take place. %
\begin{description}[leftmargin=3.5mm]

\item{(1)}
Filtering out each third element of
$\estreamthree{1}{1}{1}{1}$
vacuously constructs $\estreamthree{1}{1}{1}{1}$.
The stream of the resulting prefix sums is $\estreamthree{1}{2}{3}{4}$.

\item{(0)}
Filtering out each second element of
this last stream
constructs $\estreamthree{1}{3}{5}{7}$, the stream of the
positive odd integers.
As
mentioned %
at the bottom of page~\pageref{page:sum-of-first-odd-natural-numbers},
the stream of the resulting prefix sums is $\estreamthree{1}{4}{9}{16}$,
\ie, $\estreamthree{1^2}{2^2}{3^2}{4^2}$,
the stream of the
positive integers exponentiated with $2$.

\end{description}

\item[Exponent 3:] Three iterations take place. %
\begin{description}[leftmargin=3.5mm]

\item{(2)}
Filtering out each fourth element of
$\estreamthree{1}{1}{1}{1}$
vacuously constructs $\estreamthree{1}{1}{1}{1}$.
The stream of the resulting prefix sums is $\estreamthree{1}{2}{3}{4}$.

\item{(1)}
Filtering out each third element of
this last stream
constructs $\estreamseven{1}{2}{4}{5}{7}{8}{10}$.
The stream of the resulting prefix sums is
$\estreamseven{1}{3}{7}{12}{19}{27}{37}$.

\item{(0)}
Filtering out each second element of
this last stream
constructs
$\estreamsix{1}{7}{19}{37}{61}{91}$.
The stream of the resulting prefix sums is
$\estreamthree{1}{8}{27}{64}$,
\ie, $\estreamthree{1^3}{2^3}{3^3}{4^3}$,
the stream of the
positive integers exponentiated with $3$.

\end{description}

\item[Exponent 4:] Four iterations take place. %
\begin{description}[leftmargin=3.5mm]

\item{(3)}
Filtering out each fifth element of
$\estreamthree{1}{1}{1}{1}$
vacuously constructs $\estreamthree{1}{1}{1}{1}$.
The stream of the resulting prefix sums is $\estreamthree{1}{2}{3}{4}$.

\item{(2)}
Filtering out each fourth element of
this last stream
constructs $\estreamseven{1}{2}{3}{5}{6}{7}{9}$.
The stream of the resulting prefix sums is
$\estreamseven{1}{3}{6}{11}{17}{24}{33}$.

\item{(1)}
Filtering out each third element of
this last stream
constructs $\estreamsix{1}{3}{11}{17}{33}{43}$.
The stream of the resulting prefix sums is
${\ensuremath{[{1},\,{4},\,{15},\,{32},\,{65},\,{108},\,{175},\,{256},\,{369},\,{625},\cdots[}}$.

\item{(0)}
Filtering out each second element of
this last stream
constructs
$\estreamfive{1}{15}{65}{175}{369}$.
The stream of the resulting prefix sums is
$\estreamfour{1}{16}{81}{256}$,
\ie, $\estreamfour{1^4}{2^4}{3^4}{4^4}$,
the stream of the
positive integers exponentiated with $4$.

\end{description}

\end{description}

\noindent
In 1951~\cite{Moessner:51},
Alfred Moessner presented this process and conjectured that given a positive integer $n$,
it always ends with $\estreamthreea{1^n}{2^n}{3^n}{4^n}$.
This conjecture was then proved and is now referred to as Moessner's theorem.\footnote{
Alfred Moessner (1893--198?) was a German mathematics teacher
who, besides regular research activities~\cite{Moessner:zbmath},
contributed to recreational-mathematics journals like Scripta Mathematica (between 1936 and 1956 with 45 contributions)
and Sphinx (between 1931 and 1933 with 16 contributions and then in 1939 with 1 contribution)
with many results
-- diophantine equations, properties of integers, Pythagorean numbers, and a number of interesting identities.
He is noted, \eg, for
triplets~\cite{Moessner:SM41,Gloden:SM47}
as well as for magic squares~\cite{Moessner:SM47},
three of them solely containing prime numbers~\cite{Moessner:SM52}:
Honsberger refers to Moessner's results in arithmetic as ``gems''~\cite[page~269]{Honsberger:91}.
Both in Scripta Arithmetica and in Sphinx,
several of Moessner's ``Curiosa'' were followed by an insightful note from the Editor-in-Chief (Jekuthiel Ginsburg
and Maurice Kraitchik, resp.),
a testament of the liveliness of mathematics in the middle of the 20th century.
(Scripta Arithmetica was reporting events as they were happening at the time,
from the creation of the Journal of Symbolic Logic to the publication of Littlewood's miscellany,
which makes it a fascinating -- and sometimes heart-breaking -- reading today.)
As for his affiliation, Moessner did not mention any,
simply listing N{\"u}rnberg in 1936 in his first contribution to Scripta Arithmetica,
and then mentioning nothing until after the war, where he listed
Gunzenhausen next to his name, as he did in the article that introduced his eponymous theorem~\cite{Moessner:51}.
}

\subsection{In the 20th Century}
\label{subsec:in-the-20th-century}

Oskar Perron (1) characterized the series of integers after the first periodic elision in Moessner's process
and (2) characterized the series of integers after the next periodic elision,
given the series of integers after a periodic elision.
He expressed this series of integers in closed form and proved this closed form using mathematical induction.
Then, using a telescoping sum (like at the top of page~\pageref{page:telescoping-sum}),
he proved that the prefix sums of this series of integers compute powers,
as conjectured by Moessner~\cite{Perron:51}.

Ivan Paasche obtained what is now known as Moessner's theorem as a corollary of a more general theorem about
generating functions~\cite{Paasche:52} and then later using linear transformations~\cite{Paasche:AdM55}.
He also pointed out the connection between Moessner's process and Pascal's triangle as well as Taylor series~\cite{Paasche:MNU53-54}.

Jan van Yzeren
documented
a correspondence between numbers obtained in the course of Moessner's process
and numbers obtained in the course of Horner's algorithm~\cite{vanYzeren:AMM59}.

Hans Sali\'e characterized the output of Moessner's process
$\estreamthree{a_1^{(n)}}{a_2^{(n)}}{a_3^{(n)}}{a_4^{(n)}}$
in terms of the initial sequence
$\estreamthree{a_1^{(1)}}{a_2^{(1)}}{a_3^{(1)}}{a_4^{(1)}}$ for any exponent $n \ge 2$~\cite{Salie:52}.
This characterization, which is not solely additive and is renamed here for notational consistency, reads as follows for any index $k \ge 1$ in the stream:
\vspace{-1mm}
$$
a_k^{(n)}
=
\sum_{i=0}^{k-1}\:\sum_{j=1}^{n-1}\:a^{(1)}_{i \cdot n + j} \cdot (k - i)^{n - 1 - j} \cdot (k - 1 - j)^{j-1}
$$
\vspace{-1mm}

\noindent
Sali\'e
then pointed out that if
$a_k^{(1)}$ is $k$
then
$a_k^{(n)}$ is $k^n$,
as in Moessner's process,
and that if $a_1^{(1)}$ is $1$ and for all positive $k$, $a_{k+1}^{(1)}$ is $0$,
then
$a_k^{(n)}$ is $k^{n-2}$,
which anticipates the observation used in \sectionref{sec:introduction} for presentational purposes
and revisited in \sectionref{subsec:long-s-second-theorem} to recast Long's second theorem~\cite{Long:AMM66}.
Naturally, if for all positive $k$, $a_{k}^{(1)}$ is $1$,
then $a_k^{(n)}$ is $k^{n-1}$.

Using a generalization of Pascal's triangle and its binomial coefficients, 
Calvin Long initialized the process with a stream that follows an arithmetic progression~\cite{Long:AMM66},
as revisited in \sectionref{subsec:long-s-second-theorem}.
He extended his study by varying the period of elision~\cite{Long:FQ86},
building on Paasche's observation that
increasing this period makes the process map additions to multiplications, subtractions to quotients, and multiplications to exponentiations,
as in a slide rule~\cite{Conway-Guy:TBON96-Moessner-s-Magic,Paasche:MNU53-54,Paasche:CM54-56}.
On this basis, he used Moessner's theorem to imaginatively motivate students
in high school~\cite{Long:TMT82,Long:MG82,Slater:MG83} -- which is easier to achieve
today thanks to the On-Line Encyclopedia of Integer Sequences~\cite{OEIS}:
the students can copy-paste the prefix of a stream
and excitedly see whether it corresponds to a known sequence of integers.

Ross Honsberger presented a geometrical proof due to Karel Post~\cite{Honsberger:91}.

\subsection{In the 21st Century}
\label{subsec:in-the-21st-century}

After the turn of the century, Moessner's process became a classical exercise, and Roland
Backhouse introduced Ralf Hinze to it as such~\cite{Hinze:IFL08}.
Hinze elegantly presented a calculational proof of Moessner's theorem
and then developed a comprehensive framework for streams~\cite{Hinze:JFP11}.
Milad Niqui and Jan Rutten proved Moessner's theorem
coalgebraically~\cite{Niqui-Rutten:HOSC11}.
Robbert Krebbers, Louis Parlant, and Alexandra Silva formalized a coinduction proof of Moessner's theorem in Coq~\cite{Krebbers-al:16}.
Dexter Kozen and Alexandra Silva vastly generalized Moessner's theorem using power series
and presented an algebraic proof for it~\cite{Kozen-Silva:AMM13}, which
Mark Bickford formalized in the Nuprl proof
assistant, the first formalization of Moessner's theorem and of its proof~\cite{Bickford-al:JLAGMP22}.
In his MSc thesis~\cite{Urbak:MSc},
Peter Urbak dualized Moessner's process, generating triangles column by column instead of row by row
and singling out a collection of new properties.
He also formalized and generalized the correspondence
documented
by van Yzeren
between numbers obtained in the course of Moessner's process
and numbers obtained in the course of Horner's algorithm~\cite{vanYzeren:AMM59}.
In his BSc thesis~\cite{Treihis:BSc},
Uladzimir Treihis formalized Perron and Sali\'e's proofs
and generalized Sali\'e's result to a version of Moessner's process without striking-out phase.

The author's motivation is the same as in his contribution to Glynn Winskel's Fest\-schrift~\cite{Clausen-al:TCS14}.
To quote:

\vspace{-1mm}

\begin{quote}
  All in all, it seems to us that like Stonehenge, Moessner's theorem is
  like a mirror -- every publication about it reflects what is in the
  mind of its authors: a property, a proof technique, a corollary,
  or a showcase for a framework.\footnote{Or, like, dynamic programming?}
  Accordingly, the present article reflects what is in our mind as
  computer scientists and functional programmers: we want to program
  Moessner's [process] to understand not just how it works but also why it
  works.\footnote{At the time, the author put forward the word ``sieve''
  instead of ``process'' but this word is inaccurate since Moessner's process generates
  new numbers
  whereas a
  sieve does not.}
\end{quote}

\vspace{-3mm}

\section{Moessner's Process Without Dynamic Programming}
\label{sec:a-polynomial-rendition-of-moessner-s-process}

In Moessner's process, each of the streams enumerates a function.
For example,
$\estreamthreec{1^n}{2^n}{3^n}{4^n}$
enumerates
$\elam{x}{\expo{(x + 1)}{n}}$.
The goal of this section is to analyze Moessner's process (\sectionref{subsec:analysis})
and to recast its constitutive steps -- namely the filtering-out phase
(\sectionref{subsec:the-filtering-out-phase}) and then the prefix sums
(\sectionref{subsec:the-prefix-sums}) -- in terms of these enumerated functions
instead of in terms of the enumerating streams.
The result will be a streamless version of Moessner's process, \ie, a
version without dynamic programming
(\sectionref{subsec:moessner-s-process-streamlessly}).

\vspace{-1mm}

\subsection{Analysis}
\label{subsec:analysis}

Let
us characterize the
stream processing
that underlies Moessner's process.
The
filtering-out phase mentions each $i + 1$st element of
a stream:
the index $x$ of each second element has the property that $\emod{x}{2} = 1$,
  the index $x$ of each third element
  has the property that $\emod{x}{3} = 2$, and
  for all $j : \mathbb{N}$,
  the index $x$ of each $j + 1$st element
  has the property that $\emod{x}{(j+1)} = j$.
The corresponding
predicate
reads
$\elam{j}{\elam{x}{\emod{x}{(j + 1)} = j}}$,
witness the following truth table,

$$
\begin{array}{@{}l|ccccccccccccccl@{}}
j \protect{\verb"\"} x & 0 & 1 & 2 & 3 & 4 & 5 & 6 & 7 & 8 & 9 & 10 & 11 & 12 & 13 & \ldots
\\
\hline
1 & F & T & F & T & F & T & F & T & F & T & F & T & F & T
\\
2 & F & F & T & F & F & T & F & F & T & F & F & T & F & F
\\
3 & F & F & F & T & F & F & F & T & F & F & F & T & F & F
\\
\vdots
\end{array}
$$

\noindent
where $T$ denotes true and $F$ false.
And indeed, \eg, in the
row where $j = 2$, the index of each third
element of the stream satisfies the predicate and the other indices do
not.
The salient point of this table is that its rows, read bottom up,
characterize both the numbers filtered by Moessner's process 
and their period of elision.

\subsection{The Filtering-Out Phase}
\label{subsec:the-filtering-out-phase}

For any given positive integer $i$,
how do we
filter out each $i + 1$st element of a stream
to construct a new stream?
Put more simply,
for $f : \mathbb{N} \rightarrow \mathbb{N}$,
and given the
stream $\estreamthreea{\eapp{f}{0}}{\eapp{f}{1}}{\eapp{f}{2}}{\eapp{f}{3}}$,
\begin{itemize}[leftmargin=3.5mm,itemsep=1pt]

\item
  how can
    $\estreamseven{\eapp{f}{0}}{\eapp{f}{2}}{\eapp{f}{4}}{\eapp{f}{6}}{\eapp{f}{8}}{\eapp{f}{10}}{\eapp{f}{12}}$
  be constructed,
  where each second element has been
  filtered out?

\item
  how can
    $\estreamseven{\eapp{f}{0}}{\eapp{f}{1}}{\eapp{f}{3}}{\eapp{f}{4}}{\eapp{f}{6}}{\eapp{f}{7}}{\eapp{f}{9}}$
  be constructed,
  where each third element has been
  filtered out?

\item
  etc.

\end{itemize}

\noindent
To this end, let us define a function $\rawgee : \mathbb{N}^+ \rightarrow \mathbb{N} \rightarrow \mathbb{N}$
that is given the (positive) period of
elision
and the original input for $f$:
\begin{itemize}[leftmargin=3.5mm,itemsep=1pt]

\item
  $\estreamthree{\eapp{f}{(\gee{1}{0})}}{\eapp{f}{(\gee{1}{1})}}{\eapp{f}{(\gee{1}{2})}}{\eapp{f}{(\gee{1}{3})}}
   =
   \estreamseven{\eapp{f}{0}}{\eapp{f}{2}}{\eapp{f}{4}}{\eapp{f}{6}}{\eapp{f}{8}}{\eapp{f}{10}}{\eapp{f}{12}}$

\item
  $\estreamthree{\eapp{f}{(\gee{2}{0})}}{\eapp{f}{(\gee{2}{1})}}{\eapp{f}{(\gee{2}{2})}}{\eapp{f}{(\gee{2}{3})}}
   =
   \estreameight{\eapp{f}{0}}{\eapp{f}{1}}{\eapp{f}{3}}{\eapp{f}{4}}{\eapp{f}{6}}{\eapp{f}{7}}{\eapp{f}{9}}{\eapp{f}{10}}$

\item
  etc.

\end{itemize}

\noindent
We define $\rawgee$ so that
the arithmetic progression is maintained and a number is periodically
skipped.
The definition of $\rawgee$
hinges on
integer division (and so we omit the
floor
notation used in the executive summary since we are considering integers,
not rational numbers):

\begin{property}[basic arithmetic for natural numbers]
\label{prop:basic-arithmetic}
$
\forall j : \mathbb{N},
\forall x : \mathbb{N},
\\
\left\{
\begin{array}{r@{\ }c@{\ }l}
\emod{x}{(j+1)} < j & \Rightarrow & \frac{x+1}{j+1} = \frac{x}{j+1}
\\
\emod{x}{(j+1)} = j & \Rightarrow & \frac{x+1}{j+1} = \frac{x}{j+1} + 1
\end{array}
\right.
$
\end{property}

\label{page:elision-function}
\begin{definition}[Elision function]
\label{def:elision-function}
$\rawgee : \mathbb{N}^+ \rightarrow \mathbb{N} \rightarrow \mathbb{N}
         = \elam{j}{\elam{x}{\frac{(j + 1) \cdot x}{j}}}$
\end{definition}

\noindent
Equivalently, we could define $g : \mathbb{N} \rightarrow \mathbb{N} \rightarrow \mathbb{N}
   = \elam{j}{\elam{x}{\frac{(j + 2) \cdot x}{j+1}}}
   = \elam{j}{\elam{x}{x + \frac{x}{j+1}}}$,
which illustrates the relevance of Property~\ref{prop:basic-arithmetic}.

The following table illustrates
$\rawgee$
(in each row, $\circ$ prefixes the successor of a number that has been skipped):
$$
\begin{array}{@{}r|rrrrrrrrrrrrrrr@{}}
x & 
0 & 1 & 2 & 3 & 4 & 5 & 6 & 7 & 8 & 9 & 10 & 11 & 12 & 13 & \ldots
\\
\hline
g \: 1 \: x &
0 & \circ 2 & \circ 4 & \circ 6 & \circ 8 & \circ 10 & \circ 12 & \circ 14 & \circ 16 & \circ 18 & \circ 20 & \circ 22 & \circ 24 & \circ 26 %
\\
g \: 2 \: x &
0 & 1 & \circ 3 & 4 & \circ 6 & 7 & \circ 9 & 10 & \circ 12 & 13 & \circ 15 & 16 & \circ 18 & 19 %
\\
g \: 3 \: x &
0 & 1 & 2 & \circ 4 & 5 & 6 & \circ 8 & 9 & 10 & \circ 12 & 13 & 14 & \circ 16 & 17 %
\\
\vdots
\end{array}
$$

\noindent
And indeed,
in the row for $g \: 1 \: x$, each successive second element has been filtered out,
in the row for $g \: 2 \: x$, each successive third element has been filtered out,
etc.

It also makes sense to define the complement of $\rawgee$ -- noted $\rawgeeop$ -- explicitly,
since Moessner's process is defined in terms of the numbers that are filtered
out,
not in terms of the numbers that are filtered
in~\cite{Long:MG82,Slater:MG83}:

\begin{definition}[Complement of the elision function]
$\rawgeeop
 = \elam{j}{\elam{x}{(j + 1) \cdot x + j}}$
\end{definition}

Given a period of elision, the numbers that are skipped enumerate $\rawgeeop$:
$$
\begin{array}{@{}r|rrrrrrr@{}}
x & 
0 & 1 & 2 & 3 & 4 & 5 & \ldots
\\
\hline
\geeop{1}{x} &
1 & 3 & 5 & 7 & 9 & 11
\\
\geeop{2}{x} &
2 & 5 & 8 & 11 & 14 & 17
\\
\geeop{3}{x} &
3 & 7 & 11 & 15 & 19 & 23
\\
\vdots
\end{array}
$$

\subsection{The Prefix Sums}
\label{subsec:the-prefix-sums}

If a stream enumerates a given function $f$,
then the stream of its prefix sums enumerates $\elam{x}{\sum_{i=0}^x\,\eapp{f}{i}}$.

\subsection{Moessner's Process, Streamlessly}
\label{subsec:moessner-s-process-streamlessly}

Instead of considering the successive streams,
we consider the successive
functions that these streams enumerate.
So, given an exponent $n$, we start from the constant function $f_n =
\elam{x}{1}$ rather than from the stream of $1$'s.
The first iteration constructs $f_{n-1} = \elam{x}{\sum_{i_n=0}^x\,\eappq{f_n}{\gee{n}{i_n}}}$,
the second iteration constructs $f_{n-2} = \elam{x}{\sum_{i_{n-1}=0}^x\,\eappq{f_{n-1}}{\gee{(n - 1)}{i_{n-1}}}}$,
etc., and we shall unfold the definition of $\rawgee$ as we go.

To illustrate,
here are the renditions of Moessner's process for the exponents 0, 1, 2, 3, and 4,
an echo of \sectionref{sec:background-and-related-work}.
The initial stream is $\estreamthree{1}{1}{1}{1}$ and
we shall unfold the definitions of $f_1$, $f_2$, $f_3$, and $f_4$ as we go:

\begin{description}[leftmargin=3.5mm]

\item[Exponent 0:] %
The initial stream
enumerates $f_0 = \elam{x}{1}$ and no iteration takes place.
The result is $\elam{x}{1}$,
\ie, $\elam{x}{\expo{(x + 1)}{0}}$.

\item[Exponent 1:] %
The initial stream
enumerates $f_1 = \elam{x}{1}$ and one iteration takes place. %
\begin{description}[leftmargin=3.5mm]

\item{(0)}
The iteration constructs
$\estreamthree{1}{2}{3}{4}$.
This stream enumerates $f_0 =$ \\
$\elam{x}{\sum_{i_1=0}^x\,\eappq{f_1}{\gee{1}{i_1}}}
 = %
 \elam{x}{\sum_{i_1=0}^x\,1}
 =
 \elam{x}{(x + 1)}$,
\ie, $\elam{x}{\expo{(x + 1)}{1}}$.

\end{description}

\item[Exponent 2:] %
The initial stream
enumerates $f_2 = \elam{x}{1}$, and two iterations take place. %
\begin{description}[leftmargin=3.5mm]

\item{(1)}
The first iteration constructs
$\estreamthree{1}{2}{3}{4}$.
This stream enumerates $f_1 =$ \\
$\elam{x}{\sum_{i_2=0}^x\,\eappq{f_2}{\gee{2}{i_2}}}
 = %
 \elam{x}{\sum_{i_2=0}^x\,1}
 =
 \elam{x}{(x + 1)}
$.

\item{(0)}
The second iteration constructs
$\estreamthree{1}{4}{9}{16}$.
This stream enumerates $f_0 = $ \\
$\elam{x}{\sum_{i_1=0}^x\,\eappq{f_1}{\gee{1}{i_1}}} %
 =
 \elam{x}{\sum_{i_1=0}^x\,\eappq{f_1}{\frac{2 \cdot i_1}{1}}}
 =
 \elam{x}{\sum_{i_1=0}^x\,(\frac{2 \cdot i_1}{1} + 1)}
 =
 \elam{x}{\sum_{i_1=0}^x\,(2 \cdot i_1 + 1)}$,
\ie, $\elam{x}{\expo{(x + 1)}{2}}$
since as detailed
at the bottom of page~\pageref{page:sum-of-first-odd-natural-numbers}
and at the top of page~\pageref{page:telescoping-sum},
summing the first odd natural numbers gives a square number.

\end{description}

\item[Exponent 3:] %
The initial stream
enumerates $f_3 = \elam{x}{1}$, and three iterations take place. %
\begin{description}[leftmargin=3.5mm]

\item{(2)}
The first iteration constructs $\estreamthree{1}{2}{3}{4}$.
This stream enumerates $f_2 =$ \\
$\elam{x}{\sum_{i_3=0}^x\,\eappq{f_3}{\gee{3}{i_3}}}
 = %
 \elam{x}{\sum_{i_3=0}^x\,1}
 =
 \elam{x}{(x + 1)}$.

\item{(1)}
The second iteration constructs $\estreamseven{1}{3}{7}{12}{19}{27}{37}$.
This stream enumerates $f_1 = $ \\
$
\elam{x}{\sum_{i_2=0}^x\,\eappq{f_2}{\gee{2}{i_2}}}
 = 
 \elam{x}{\sum_{i_2=0}^x\,\eappq{f_2}{\frac{3 \cdot i_2}{2}}}
 =
 \elam{x}{\sum_{i_2=0}^x\,(\frac{3 \cdot i_2}{2} + 1)}$.

\item{(0)}
The third iteration constructs $\estreamfour{1}{8}{27}{64}$.
This stream enumerates $f_0 =$ \\
$\elam{x}{\sum_{i_1=0}^x\,\eappq{f_1}{\gee{1}{i_1}}} = %
 \elam{x}{\sum_{i_1=0}^x\,\eappq{f_1}{\frac{2 \cdot i_1}{1}}}
 =
 \elam{x}{\sum_{i_1=0}^x\,\sum_{i_2=0}^{2 \cdot i_1}\,(\frac{3 \cdot i_2}{2} + 1)}$,
and \\
$\sum_{i_1=0}^x\,\sum_{i_2=0}^{2 \cdot i_1}\,(\frac{3 \cdot i_2}{2} + 1) =%
\sum_{i_1=0}^x\,\sum_{i_2=0}^{2 \cdot i_1}\,(i_2 + \frac{i_2}{2} + 1) = $ \\
$
\sum_{i_1=0}^x\,(\sum_{i_2=0}^{2 \cdot i_1}\,i_2 + \sum_{i_2=0}^{2 \cdot i_1}\,\frac{i_2}{2} + \sum_{i_2=0}^{2 \cdot i_1}\,1) =%
\sum_{i_1=0}^x\,(\frac{(2 \cdot i_1) \cdot (2 \cdot i_1 + 1)}{2} + i_1^2 + (2 \cdot i_1 + 1)) = $ 
(see \appendixref{app:a-curiosa-about-polygonal-numbers}) \\
$
\sum_{i_1=0}^x\,(3 \cdot i_1^2 + 3 \cdot i_1 + 1) =%
\sum_{i_1=0}^x\,((i_1 + 1)^3 - i_1^3) =%
(x + 1)^3$

\end{description}

\label{page:the-origin-of-the-curiosa-about-polygonal-numbers}

\item[Exponent 4:] %
The initial stream
enumerates $f_4 = \elam{x}{1}$, and four iterations take place. %
\begin{description}[leftmargin=3.5mm]

\item{(3)}
The first iteration constructs $\estreamthree{1}{2}{3}{4}$.
This stream enumerates $f_3 =$ \\
$\elam{x}{\sum_{i_4=0}^x\,\eappq{f_4}{\gee{4}{i_4}}}
 =
 \elam{x}{\sum_{i_4=0}^x\,1}$.

\item{(2)}
The second iteration constructs $\estreamseven{1}{3}{6}{11}{17}{24}{33}$.
This stream enumerates
$f_2
 = $ \\
$
 \elam{x}{\sum_{i_3=0}^x\,\eappq{f_3}{\gee{3}{i_3}}}
 =
 \elam{x}{\sum_{i_3=0}^x\,\eappq{f_3}{\frac{4 \cdot i_3}{3}}}$.

\item{(1)}
The third iteration constructs $\estreameight{1}{4}{15}{32}{65}{108}{175}{256}$.
This stream enumerates
$f_1
 = $ \\
$
 \elam{x}{\sum_{i_2=0}^x\,\eappq{f_2}{\gee{2}{i_2}}}
 =
 \elam{x}{\sum_{i_2=0}^x\,\eappq{f_2}{\frac{3 \cdot i_2}{2}}}$.

\item{(0)}
The fourth iteration constructs $\estreamfour{1}{16}{81}{256}$.
This stream enumerates $f_0 =$ \\
$\elam{x}{\sum_{i_1=0}^x\,\eappq{f_1}{\gee{1}{i_1}}}
 =
 \elam{x}{\sum_{i_1=0}^x\,\eappq{f_1}{\frac{2 \cdot i_1}{1}}}
 =
 \ldots
 =
 \elam{x}{\sum_{i_1=0}^x\,((i_1 + 1)^4 - i_1^4)}
$,
\ie, $\elam{x}{(x + 1)^4}$.

\end{description}

\end{description}

Let us unfold all the functions that are enumerated by the intermediate streams
for the exponents 0, 1, 2, 3, and 4 as well as 5:
$$
\begin{array}{l}
\displaystyle
\elam{x}{1}
\\[2mm]
\displaystyle
\elam{x}{\sum_{i_1=0}^x\:1}
\\[2mm]
\displaystyle
\elam{x}{\sum_{i_1=0}^x\:\sum_{i_2=0}^{\gee{1}{i_1}}\:1}
\\[2mm]
\displaystyle
\elam{x}{\sum_{i_1=0}^x\:\sum_{i_2=0}^{\gee{1}{i_1}}\:\sum_{i_3=0}^{\gee{2}{i_2}}\:1}
\\[2mm]
\displaystyle
\elam{x}{\sum_{i_1=0}^x\:\sum_{i_2=0}^{\gee{1}{i_1}}\:\sum_{i_3=0}^{\gee{2}{i_2}}\:\sum_{i_4=0}^{\gee{3}{i_3}}\:1}
\\[2mm]
\displaystyle
\elam{x}{\sum_{i_1=0}^x\:\sum_{i_2=0}^{\gee{1}{i_1}}\:\sum_{i_3=0}^{\gee{2}{i_2}}\:\sum_{i_4=0}^{\gee{3}{i_3}}\:\sum_{i_5=0}^{\gee{4}{i_4}}\:1}
\end{array}
$$

\noindent
Each nested sum is an instance of a prefix sum and each upper bound that
involves $\rawgee$ is an instance of the filtering-out phase.
(In \sectionref{subsec:an-alternative-theorem-to-obtain-powers} and onwards,
we consider other definitions of $\rawgee$, \ie, other filtering-out strategies.)

More generally, unfolding all the intermediate definitions leads one to
the following
streamless instances of Moessner's theorem for the exponents 0, 1, 2, 3, and
4,
without the dynamic-programming infrastructure:
\begin{eqnarray*}
\elam{x}{(x + 1)^0}
& = &
\elam{x}{1}
\\[1mm]
\elam{x}{(x + 1)^1}
& = &
\elam{x}{\sum_{i_1=0}^x\,1}
\\[-1mm]
\elam{x}{(x + 1)^2}
& = &
\elam{x}{\sum_{i_1=0}^x\,\sum_{i_2=0}^{\frac{2 \cdot i_1}{1}}\,1}
\hspace{1.5cm}\;\;=\;\;
\elam{x}{\sum_{i_1=0}^x\,\sum_{i_2=0}^{i_1 + \frac{i_1}{1}}\,1}
\\[-1mm]
\elam{x}{(x + 1)^3}
& = &
\elam{x}{\sum_{i_1=0}^x\,\sum_{i_2=0}^{\frac{2 \cdot i_1}{1}}\,\sum_{i_3=0}^{\frac{3 \cdot i_2}{2}}\,1}
\hspace{0.75cm}\;\;=\;\;
\elam{x}{\sum_{i_1=0}^x\,\sum_{i_2=0}^{i_1 + \frac{i_1}{1}}\,\sum_{i_3=0}^{i_2 + \frac{i_2}{2}}\,\,1}
\\[-1mm]
\elam{x}{(x + 1)^4}
& = &
\elam{x}{\sum_{i_1=0}^x\,\sum_{i_2=0}^{\frac{2 \cdot i_1}{1}}\,\sum_{i_3=0}^{\frac{3 \cdot i_2}{2}}\,\sum_{i_4=0}^{\frac{4 \cdot i_3}{3}}\,1}
\;\;=\;\;
\elam{x}{\sum_{i_1=0}^x\,\sum_{i_2=0}^{i_1 + \frac{i_1}{1}}\,\sum_{i_3=0}^{i_2 + \frac{i_2}{2}}\,\sum_{i_4=0}^{i_3 + \frac{i_3}{3}}\,1}
\end{eqnarray*}

\section{A Left Inverse of Moessner's Process Without Dynamic Programming}
\label{sec:a-polynomial-rendition-of-the-left-inverse-of-moessner-s-process}

Connecting prefix sums (antidifferences) and backward differences makes it possible
to devise a left inverse for Moessner's process.
This left inverse has no computational interest in terms of its final result
-- which is $\estreamthree{1^0}{2^0}{3^0}{4^0}$, \ie, $\estreamthree{1}{1}{1}{1}$,
for any given exponent $n$ and stream $\estreamthree{1^n}{2^n}{3^n}{4^n}$ --
but studying it provides
a second
way to construct the intermediate streams.
The fun is in the game but the reader in a hurry can safely skip this consolidating section.

\subsection{Inverting the Prefix Sums}

We use backward differences, since as shown in \sectionref{subsec:prerequisites-and-notations},
they invert prefix sums.
  When going from $\estreamthree{1^0}{2^0}{3^0}{4^0}$ to $\estreamthree{1^n}{2^n}{3^n}{4^n}$ 
in Moessner's process,
the exponents increase since each iteration of the process uses an antidifference,
which is the discrete counterpart of an integration.
And when going from $\estreamthree{1^n}{2^n}{3^n}{4^n}$ to $\estreamthree{1^0}{2^0}{3^0}{4^0}$
in the converse of Moessner's process,
the exponents decrease since each iteration of the process uses a difference,
which is the discrete counterpart of a
differentiation.

\subsection{Inverting the Filtering-Out Phase}
\label{subsec:inverting-the-filtering-out-phase}

To invert the filtering-out phase, we use a converse of the elision function from
\sectionref{subsec:the-filtering-out-phase}.
The following table illustrates the idea of
this converse
with
$$h =
\elam
  {j}
  {\elam
     {x}
     {\econditional
       {\emod{x}{(j+1)} = j}
       {\square}
       {j \cdot \frac{x}{j + 1}}}}$$
where $\square$ is a placeholder for the stream element to insert and using
the converse of $\frac{(j + 1) \cdot x}{j}$,
which yields the following ``staircase'' table:

$$
\begin{array}{@{}r|rrrrrrrrrrrrrrrrr@{}}
x & 
0 & 1 & 2 & 3 & 4 & 5 & 6 & 7 & 8 & 9 & 10 & 11 & 12 & 13 & 14 & 15 & \ldots
\\
\hline
h \: 1 \: x &
0 & \square & 1 & \square & 2 & \square & 3 & \square & 4 & \square & 5 & \square & 6 & \square & 7 & \square
\\
h \: 2 \: x &
0 & 0 & \square & 2 & 2 & \square & 4 & 4 & \square & 6 & 6 & \square & 8 & 8 & \square & 10
\\
h \: 3 \: x &
0 & 0 & 0 & \square & 3 & 3 & 3 & \square & 6 & 6 & 6 & \square & 9 & 9 & 9 & \square
\\
h \: 4 \: x &
0 & 0 & 0 & 0 & \square & 4 & 4 & 4 & 4 & \square & 8 & 8 & 8 & 8 & \square & 12
\\
\vdots
\end{array}
$$

\noindent
And indeed:
\begin{itemize}[leftmargin=3.5mm,itemsep=1pt]

\item~\hbox{in the row for $h \: 1 \: x$, each stair has size 1 and every 2nd cell, there is a placeholder between each stair;}

\item
  in the row for $h \: 2 \: x$, each stair has size 2 and every 3rd cell, there is a placeholder between each stair,

\item
  in the row for $h \: 3 \: x$, each stair has size 3 and every 4th cell, there is a placeholder between each stair,

\item
  etc.

\end{itemize}

The converse of the filtering-out phase is achieved
by replacing the stairs in the table above (\eg, 0, 0, 0 and 3, 3, 3 in the row for $h \: 3 \: x$)
by increasing sequences (namely 0, 1, 2 and 3, 4, 5 in this row).
These stairs are generated by the alternative branch,
$j \cdot \frac{x}{j + 1}$, \ie, when $\emod{x}{(j+1)} \neq j$ in the definition of $h$.
So
we
add $\emod{x}{(j+1)}$
in that
branch,
which
yields
$k$ $=$
\\
$\elam
  {j}
  {\elam
     {x}
     {\econditional
       {\emod{x}{(j+1)} = j}
       {\square}
       {j \cdot \frac{x}{j + 1} + \emod{x}{(j + 1)}}}}$:
$$
\begin{array}{@{}r|rrrrrrrrrrrrrrrrc@{}}
x & 
0 & 1 & 2 & 3 & 4 & 5 & 6 & 7 & 8 & 9 & 10 & 11 & 12 & 13 & 14 & 15 & \ldots
\\
\hline
k \: 1 \: x &
0 & \square & 1 & \square & 2 & \square & 3 & \square & 4 & \square & 5 & \square & 6 & \square & 7 & \square
\\
k \: 2 \: x &
0 & 1 & \square & 2 & 3 & \square & 4 & 5 & \square & 6 & 7 & \square & 8 & 9 & \square & 10
\\
k \: 3 \: x &
0 & 1 & 2 & \square & 3 & 4 & 5 & \square & 6 & 7 & 8 & \square & 9 & 10 & 11 & \square
\\
k \: 4 \: x &
0 & 1 & 2 & 3 & \square & 4 & 5 & 6 & 7 & \square & 8 & 9 & 10 & 11 & \square & 12
\\
\vdots
\end{array}
$$

\noindent
And indeed,
each row contains the successive natural numbers if one skips the placeholders,
which provides the indices
for
each of the intermediate streams
in the converse of Moessner's process.

Earlier on~\cite{Clausen-al:TCS14}, the author observed that in Moessner's process,
the parts of the stream that are filtered out
enumerate the successive monomials of the binomial expansion of $\expo{((x + 1) + 1)}{n}$ --
which explains why the elements standing in the resulting stream enumerate $\expo{(x + 1)}{n}$,
an alternative proof of Moessner's theorem revisited in \sectionref{sec:the-essence-of-moessner-s-theorem}.
So we splice in these successive monomials in the place of $\square$ in the table above
to invert the filtering-out phase.

\subsection{Inverting Moessner's Process, Streamlessly}
\label{subsec:inverting-moessner-s-process-streamlessly}

Here are the renditions of the left inverse of Moessner's process for the exponents 2, 3, and 4,
as an echo of
\sectionref{subsec:moessner-s-process-streamlessly} and
where $\omega_k = \elam{x}{\expo{(x + 1)}{k}}$:

\begin{description}[leftmargin=3.5mm]

\item[Exponent 2:]
The initial stream, $\estreamthree{1}{4}{9}{16}$, enumerates $f_0 = \omega_2$,
\ie, $\elam{x}{\expo{(x + 1)}{2}}$,
or again \\
$\elam{x}{\binom{2}{2} \cdot (\frac{x}{1} + 1)^2}$,
and two iterations take place. %
\begin{description}[leftmargin=3.5mm]

\item{(0)}
The first iteration constructs the stream $\estreamthree{1}{2}{3}{4}$.
This stream enumerates \\
$f_1
 =
 {\elam
    {x}
    {\econditional
       {\emod{x}{2} = 1}
       {\binom{2}{1} \cdot (\frac{x}{2} + 1)^1}
       {\edifference{f_0}{(1 \cdot \frac{x}{2} + \emod{x}{2})}}}}$.

\item{(1)}
The second iteration constructs the stream $\estreamthree{1}{1}{1}{1}$.
This stream enumerates \\
$f_2
 =
 {\elam
    {x}
    {\econditional
       {\emod{x}{3} = 2}
       {\binom{2}{0} \cdot (\frac{x}{3} + 1)^0}
       {\edifference{f_1}{(2 \cdot \frac{x}{3} + \emod{x}{3})}}}}$.

\end{description}

Each $f_i$ is equivalent to the corresponding $f_i$ for exponent 2
in \sectionref{subsec:moessner-s-process-streamlessly}
and $f_2$ is equivalent to
$\omega_0$, \ie, $\elam{x}{\expo{(x + 1)}{0}}$.

\item[Exponent 3:]
The initial stream, $\estreamfour{1}{8}{27}{64}$, enumerates $f_0 = \omega_3$,
\ie, $\elam{x}{\expo{(x + 1)}{3}}$, or again \\
$\elam{x}{\binom{3}{3} \cdot (\frac{x}{1} + 1)^3}$,
and three iterations take place. %
\begin{description}[leftmargin=3.5mm]

\item{(0)}
The first iteration constructs the stream
$\estreamfive{1}{3}{7}{12}{19}$.
This stream enumerates \\
$f_1
 =
 {\elam
    {x}
    {\econditional
       {\emod{x}{2} = 1}
       {\binom{3}{2} \cdot (\frac{x}{2} + 1)^2}
       {\edifference{f_0}{(1 \cdot \frac{x}{2} + \emod{x}{2})}}}}$.

\item{(1)}
The second iteration constructs the stream $\estreamthree{1}{2}{3}{4}$.
This stream enumerates \\
$f_2
 =
 {\elam
    {x}
    {\econditional
       {\emod{x}{3} = 2}
       {\binom{3}{1} \cdot (\frac{x}{3} + 1)^1}
       {\edifference{f_1}{(2 \cdot \frac{x}{3} + \emod{x}{3})}}}}$.

\item{(2)}
The third iteration constructs the stream $\estreamthree{1}{1}{1}{1}$.
This stream enumerates \\
$f_3
 =
 {\elam
    {x}
    {\econditional
       {\emod{x}{4} = 3}
       {\binom{3}{0} \cdot (\frac{x}{4} + 1)^0}
       {\edifference{f_2}{(3 \cdot \frac{x}{4} + \emod{x}{4})}}}}$.

\end{description}

Each $f_i$ is equivalent to the corresponding $f_i$ for exponent 3
in \sectionref{subsec:moessner-s-process-streamlessly}
and $f_3$ is equivalent to
$\omega_0$, \ie, $\elam{x}{\expo{(x + 1)}{0}}$.

\item[Exponent 4:]
The initial stream, $\estreamfour{1}{16}{81}{256}$, enumerates $f_0 = \omega_4$,
\ie, $\elam{x}{\expo{(x + 1)}{4}}$, or again $\elam{x}{\binom{4}{4} \cdot (\frac{x}{1} + 1)^4}$,
and four iterations take place. %
\begin{description}[leftmargin=3.5mm]

\item{(0)}
The first iteration constructs the stream
$\estreamfour{1}{4}{15}{32}$.
This stream enumerates \\
$f_1
 =
 {\elam
    {x}
    {\econditional
       {\emod{x}{2} = 1}
       {\binom{4}{3} \cdot (\frac{x}{2} + 1)^3}
       {\edifference{f_0}{(1 \cdot \frac{x}{2} + \emod{x}{2})}}}}$.

\item{(1)}
The second iteration constructs the stream
$\estreamfour{1}{3}{6}{11}$
This stream enumerates \\
$f_2
 =
 {\elam
    {x}
    {\econditional
       {\emod{x}{3} = 2}
       {\binom{4}{2} \cdot (\frac{x}{3} + 1)^2}
       {\edifference{f_1}{(2 \cdot \frac{x}{3} + \emod{x}{3})}}}}$.

\item{(2)}
The third iteration constructs the stream
$\estreamthree{1}{2}{3}{4}$
This stream enumerates \\
$f_3
 =
 {\elam
    {x}
    {\econditional
       {\emod{x}{4} = 3}
       {\binom{4}{1} \cdot (\frac{x}{4} + 1)^1}
       {\edifference{f_2}{(3 \cdot \frac{x}{4} + \emod{x}{4})}}}}$.

\item{(3)}
The fourth iteration constructs the stream
$\estreamthree{1}{1}{1}{1}$.
This stream enumerates \\
$f_4
 =
 {\elam
    {x}
    {\econditional
       {\emod{x}{5} = 4}
       {\binom{4}{0} \cdot (\frac{x}{5} + 1)^0}
       {\edifference{f_3}{(4 \cdot \frac{x}{5} + \emod{x}{5})}}}}$.

\end{description}

Each $f_i$ is equivalent to the corresponding $f_i$ for exponent 4
in \sectionref{subsec:moessner-s-process-streamlessly}
and $f_4$ is equivalent to
$\omega_0$, \ie, $\elam{x}{\expo{(x + 1)}{0}}$.

\end{description}

All in all, we now have two ways to carry out the
iterations
of Moessner's process,
namely forward and backward.
These two ways consolidate our understanding of it.

\section{The Essence of Moessner's Theorem}
\label{sec:the-essence-of-moessner-s-theorem}

The key observation in
\sectionref{subsec:inverting-the-filtering-out-phase} was that
Moessner's process hinges on a family of equalities indexed by the exponent.
Here are members of this family for the exponents 0, 1, 2, 3, and 4:
\begin{description}[leftmargin=3.5mm]

\item[Exponent 0:] %
$\begin{array}[t]{@{}r@{\ }c@{\ }l}
\forall i_0 : \mathbb{N},
1
& = &
\binom{0}{0} \cdot (i_0 + 1)^0
\end{array}$

\item[Exponent 1:] %
$\begin{array}[t]{@{}r@{\ }c@{\ }l}
\forall i_1 : \mathbb{N},
1
& = &
\binom{1}{0} \cdot (i_1 + 1)^0
\\[2mm]
\forall i_0 : \mathbb{N},
\sum_{i_1 = 0}^{i_0}\:
1
& = &
\binom{1}{1} \cdot (i_0 + 1)^1
\end{array}$

\item[Exponent 2:] %
$\begin{array}[t]{@{}r@{\ }c@{\ }l}
\forall i_2 : \mathbb{N},
1
& = &
\binom{2}{0} \cdot (i_2 + 1)^0
\\[2mm]
\forall i_1 : \mathbb{N},
\sum_{i_2 = 0}^{i_1}\:
1
& = &
\binom{2}{1} \cdot (i_1 + 1)^1
\\[2mm]
\forall i_0 : \mathbb{N},
\sum_{i_1 = 0}^{i_0}\:
\sum_{i_2 = 0}^{\geen{2}{1}{i_1}}\:
1
& = &
\binom{2}{2} \cdot (i_0 + 1)^2
\end{array}$

\item[Exponent 3:] %
$\begin{array}[t]{@{}r@{\ }c@{\ }l}
\forall i_3 : \mathbb{N},
1
& = &
\binom{3}{0} \cdot (i_3 + 1)^0
\\[2mm]
\forall i_2 : \mathbb{N},
\sum_{i_3 = 0}^{i_2}\:
1
& = &
\binom{3}{1} \cdot (i_2 + 1)^1
\\[2mm]
\forall i_1 : \mathbb{N},
\sum_{i_2 = 0}^{i_1}\:
\sum_{i_3 = 0}^{\geen{3}{2}{i_2}}\:
1
& = &
\binom{3}{2} \cdot (i_1 + 1)^2
\\[2mm]
\forall i_0 : \mathbb{N},
\sum_{i_1 = 0}^{i_0}\:
\sum_{i_2 = 0}^{\geen{3}{1}{i_1}}\:
\sum_{i_3 = 0}^{\geen{3}{2}{i_2}}\:
1
& = &
\binom{3}{3} \cdot (i_0 + 1)^3
\end{array}$

\item[Exponent 4:] %
$\begin{array}[t]{@{}r@{\ }c@{\ }l}
\forall i_4 : \mathbb{N},
1
& = &
\binom{4}{0} \cdot (i_4 + 1)^0
\\[2mm]
\forall i_3 : \mathbb{N},
\sum_{i_4 = 0}^{i_3}\:
1
& = &
\binom{4}{1} \cdot (i_3 + 1)^1
\\[2mm]
\forall i_2 : \mathbb{N},
\sum_{i_3 = 0}^{i_2}\:
\sum_{i_4 = 0}^{\geen{3}{3}{i_3}}\:
1
& = &
\binom{4}{2} \cdot (i_2 + 1)^2
\\[2mm]
\forall i_1 : \mathbb{N},
\sum_{i_2 = 0}^{i_1}\:
\sum_{i_3 = 0}^{\geen{3}{2}{i_2}}\:
\sum_{i_4 = 0}^{\geen{3}{3}{i_3}}\:
1
& = &
\binom{4}{3} \cdot (i_1 + 1)^3
\\[2mm]
\forall i_0 : \mathbb{N},
\sum_{i_1 = 0}^{i_0}\:
\sum_{i_2 = 0}^{\geen{3}{1}{i_1}}\:
\sum_{i_3 = 0}^{\geen{3}{2}{i_2}}\:
\sum_{i_4 = 0}^{\geen{3}{3}{i_3}}\:
1
& = &
\binom{4}{4} \cdot (i_0 + 1)^4
\end{array}$

\end{description}

\noindent
where the elision function $\rawgee$ was defined at the bottom of page~\pageref{page:elision-function} in \sectionref{subsec:the-filtering-out-phase}:
$\rawgee = \elam{j}{\elam{x}{\frac{(j + 1) \cdot x}{j}}} = \elam{j}{\elam{x}{x + \frac{x}{j}}}$.
Ostensibly,
the type of $\rawgee$
is
$\mathbb{N}^+ \rightarrow \mathbb{N} \rightarrow \mathbb{N}$,
but in actuality,
and as illustrated just above,
its
domain is $[1\;..\;n-1]$,
where $n$ is the given exponent.
In words: starting from $1$, a series of summations is performed that
periodically filters in numbers with a period that starts at
the desired exponent $n$ and that shrinks at each iteration until it
reaches $0$.
And at each iteration, the numbers that are filtered out sum up to a monomial
in the binomial expansion of $\expo{((x + 1) + 1)}{n}$,
where $x$ is the
upper bound
in each outer summation.

Moessner's theorem is the last equality in each member of this family of equalities:

\begin{description}[leftmargin=3.5mm,itemsep=1pt]

\item[Exponent 0:]
$
\forall x : \mathbb{N},
1
\; = \;  %
(x + 1)^0
$

\item[Exponent 1:]
$
\forall x : \mathbb{N},
\sum_{i_1 = 0}^x\:
1
\; = \;  %
(x + 1)^1
$

\item[Exponent 2:]
$
\forall x : \mathbb{N},
\sum_{i_1 = 0}^x\:
\sum_{i_2 = 0}^{\frac{2 \cdot i_1}{1}}\:
1
\; = \;  %
(x + 1)^2
$

\item[Exponent 3:] 
$
\forall x : \mathbb{N},
\sum_{i_1 = 0}^x\:
\sum_{i_2 = 0}^{\frac{2 \cdot i_1}{1}}\:
\sum_{i_3 = 0}^{\frac{3 \cdot i_2}{2}}\:
1
\; = \;  %
(x + 1)^3
$

\item[Exponent 4:] 
$
\forall x : \mathbb{N},
\sum_{i_1 = 0}^x\:
\sum_{i_2 = 0}^{\frac{2 \cdot i_1}{1}}\:
\sum_{i_3 = 0}^{\frac{3 \cdot i_2}{2}}\:
\sum_{i_4 = 0}^{\frac{4 \cdot i_3}{3}}\:
1
\; = \;  %
(x + 1)^4
$

\item[Exponent 5:] 
$
\forall x : \mathbb{N},
\sum_{i_1 = 0}^x\:
\sum_{i_2 = 0}^{i_1 + \frac{i_1}{1}}\:
\sum_{i_3 = 0}^{i_2 + \frac{i_2}{2}}\:
\sum_{i_4 = 0}^{i_3 + \frac{i_3}{3}}\:
\sum_{i_5 = 0}^{i_4 + \frac{i_4}{4}}\:
1
\; = \;  %
(x + 1)^5
$

\item[\vdots] 

\end{description}

\noindent
And so independently of dynamic programming, the essence of Moessner's
theorem is nested summations (as many nested summations as the degree of the result)
with fractionally increasing
upper bounds:

\begin{theorem}[Moessner's theorem without dynamic programming]
$\forall x : \mathbb{N},
 \forall n : \mathbb{N}$,
\begin{eqnarray*}
  \sum_{i_1=0}^x\:
  \sum_{i_2=0}^{\frac{2 \cdot i_1}{1}}\:
  \sum_{i_3=0}^{\frac{3 \cdot i_2}{2}}\:
  \cdots
  \sum_{i_n=0}^{\frac{n \cdot i_{n-1}}{n-1}}\:
  1
  \;\;=\;\;
  \sum_{i_1=0}^x\,
  \sum_{i_2=0}^{i_1 + \frac{i_1}{1}}\:
  \sum_{i_3=0}^{i_2 + \frac{i_2}{2}}\:
  \cdots
  \sum_{i_n=0}^{i_{n-1} + \frac{i_{n-1}}{n-1}}\:
  1
  & = &
  (x + 1)^n
\end{eqnarray*}
\end{theorem}

These nested summations involve many overlapping subcomputations.
For example,
$\sum_{i_1 = 0}^x\:\sum_{i_2 = 0}^{2 \cdot i_1}\:1
 =
 \sum_{i_1 = 0}^{x - 2}\:\sum_{i_2 = 0}^{2 \cdot i_1}\:1
 +
 \sum_{i_2 = 0}^{2 \cdot (x - 1)}\:1
 +
 \sum_{i_2 = 0}^{2 \cdot x}\:1
$ when $x \ge 2$.
Since
$\sum_{i_2 = 0}^{2 \cdot x}\:1
 =
 \sum_{i_2 = 0}^{2 \cdot x - 2}\:1
 +
 1
 +
 1
 =
 \sum_{i_2 = 0}^{2 \cdot (x - 1)}\:1
 +
 1
 +
 1
$,
the expression
$\sum_{i_2 = 0}^{2 \cdot (x - 1)}\:1$
is independently computed several times,
an opportunity for dynamic programming
that \sectionref{sec:back-to-dynamic-programming} revisits and quantifies.

In Moessner's theorem, the upper bound of each inner summation is a fractional increment of the current index.
It is instructive to tune these upper bounds, \ie, to try other filtering-out strategies~\cite{Danvy:JENSFEST24}:

\begin{itemize}[leftmargin=3.5mm]

\item
  How about not incrementing the current index?
  \begin{corollary}[Binomial coefficients]
  \label{coro:binomial-coefficients}
  $\forall x : \mathbb{N},
   \forall n : \mathbb{N}$,
  \begin{eqnarray*}
    \sum_{i_1=0}^x\,
    \sum_{i_2=0}^{i_1}\,
    \sum_{i_3=0}^{i_2}\,
    \cdots
    \sum_{i_n=0}^{i_{n-1}}\,
    1
    & = &
    \binom{x + n}{n}
  \end{eqnarray*}
  \end{corollary}

\item
  How about incrementing the current index with $1$?
  \begin{corollary}[Catalan numbers]
  $\forall n : \mathbb{N}$,
  \begin{eqnarray*}
    \sum_{i_1=0}^0\,
    \sum_{i_2=0}^{i_1 + 1}\,
    \sum_{i_3=0}^{i_2 + 1}\,
    \cdots
    \sum_{i_n=0}^{i_{n-1} + 1}\,
    1
    & = &
    \frac{\binom{2 \cdot n}{n}}{n + 1}
    \:\, = \;\,
    C_n
  \end{eqnarray*}
  \end{corollary}

  Generalizing the outer upper bound from $0$ to a non-negative integer $x$
  gives rise -- according to the On-Line Encyclopedia of Integer Sequences~\cite{OEIS} --
  to the $x + 1$st convolution of Catalan numbers:
  \begin{corollary}[Catalan numbers, convolved]
  $\forall x : \mathbb{N},
   \forall n : \mathbb{N}$,
  \begin{eqnarray*}
    \sum_{i_1=0}^x\,
    \sum_{i_2=0}^{i_1 + 1}\,
    \sum_{i_3=0}^{i_2 + 1}\,
    \cdots
    \sum_{i_n=0}^{i_{n-1} + 1}\,
    1
    & = &
    \frac
      {(x + 1) \cdot \binom{2 \cdot n + x}{n}}
      {n + x + 1}
  \end{eqnarray*}
  \end{corollary}

\item
  How about adding the current index to all the previous indices?
  \begin{corollary}[A125860]
  $\forall x : \mathbb{N},
   \forall n : \mathbb{N}$,
  \begin{eqnarray*}
    \sum_{i_1=0}^x\,
    \sum_{i_2=0}^{x + i_1}\,
    \sum_{i_3=0}^{x + i_1 + i_2}\,
    \cdots
    \sum_{i_n=0}^{x + i_1 + i_2 + \ldots + i_{n-1}}\,
    1
    & &
  \end{eqnarray*}
  yields the integer at Column $x$ and row $n$ in Table A125860
  in the On-Line Encyclopedia of Integer Sequences.
  \end{corollary}

\item
  How about multiplying the current index with all the previous indices?
  \begin{corollary}[Positive integers]
  $\forall n : \mathbb{N}$,
  \begin{eqnarray*}
    \sum_{i_1=0}^1\,
    \sum_{i_2=0}^{1 \cdot i_1}\,
    \sum_{i_3=0}^{1 \cdot i_1 \cdot i_2}\,
    \cdots
    \sum_{i_n=0}^{1 \cdot i_1 \cdot i_2 \ldots \cdot i_{n-1}}\,
    1
    & = &
    n + 1
  \end{eqnarray*}
  \end{corollary}

\item
  How about emulating Fibonacci numbers?
  \begin{corollary}[A137273]
  $\forall n : \mathbb{N}$,
  \begin{eqnarray*}
    \sum_{i_1=0}^0\,
    \sum_{i_2=0}^1\,
    \sum_{i_3=0}^{i_1 + i_2}\,
    \sum_{i_4=0}^{i_2 + i_3}\,
    \cdots
    \sum_{i_n=0}^{i_{n-2} + i_{n-1}}\,
    1
    & &
  \end{eqnarray*}
  yields the number of partitions of a Fibonacci number into Fibonacci parts
  according to the On-Line Encyclopedia of Integer Sequences.
  \end{corollary}

\end{itemize}

All these corollaries were discovered by playing with parameterized versions of Moessner's process without dynamic programming.
\sectionref{sec:the-essence-of-moessner-s-process} presents an unparameterized version
and
\sectionref{sec:a-parameterized-implementation-of-moessner-s-process} presents several parameterized versions.

\section{The Essence of Moessner's Process}
\label{sec:the-essence-of-moessner-s-process}

We are now in position to state the essence of Moessner's process with a
summation function and without dynamic programming.
Here is a recursive implementation of this summation function:
\inputscheme{SIGMA_REC}

\noindent
For example, given a function $f$ implemented by a procedure denoted by \inlinescheme{f},
evaluating \inlinescheme{(Sigma_rec 3 f)} gives rise to evaluating \inlinescheme{(+ (+ (+ (+ 0 (f 0)) (f 1)) (f 2)) (f 3))}.

In practice, we implement summation iteratively with an equivalent
tail-recursive procedure that uses an accumulator:
\inputscheme{SIGMA}

\noindent
Given a function $f$ implemented by a procedure denoted by \inlinescheme{f},
evaluating \inlinescheme{(Sigma 3 f)} gives rise to evaluating \inlinescheme{(+ (f 0) (+ (f 1) (+ (f 2) (f 3))))},
where addition is re-as\-so\-ci\-ated and no addition to $0$ occurs~\cite{Danvy:JFP23}.

Here is an
implementation of $\rawgee$:
\inputscheme{G}

And here is a recursive implementation of Moessner's process without dynamic programming:
\inputscheme{MOESSNER}

We then recast Moessner's theorem to characterize the output of
Moessner's process without dynamic programming:

\begin{theorem}[Moessner's theorem without dynamic programming]
For all $n : \mathbb{N}$ denoted by \mbox{\inlinescheme{n}}
and
for all $x : \mathbb{N}$ denoted by \mbox{\inlinescheme{x}},
evaluating
\vspace{-1mm}
\begin{center}
\texttt{\inlinescheme{(moessner x n)}}
\end{center}
\vspace{-1mm}
yields $(x + 1)^n$.
\end{theorem}

\noindent
In other words, the function implemented by the Scheme procedure
\vspace{-1mm}
\begin{center}
\inlinescheme{(lambda (x) (moessner x n))}
\end{center}
\vspace{-1mm}
is enumerated by $\estreamthree{1^n}{2^n}{3^n}{4^n}$.

\section{Parameterized Implementations of Moessner's Process}
\label{sec:a-parameterized-implementation-of-moessner-s-process}

Towards recasting Long's theorems
(\sectionsref{subsec:long-s-first-theorem}{subsec:long-s-second-theorem}),
let us parameterize Moessner's process with a function to apply to the current variable in the base case:
\inputscheme{MOESSNER-INIT}

\noindent
Compared to \inlinescheme{moessner}, \inlinescheme{moessner-init} takes
an extra \inlinescheme{init} argument and applies it to the current
\inlinescheme{x} in the base case.

\subsection{Moessner's Theorem}
\label{subsec:moessner-s-theorem}

Instantiating \inlinescheme{init} with the constant function that yields
\inlinescheme{1} gives the same exponentiation as in
\sectionref{sec:the-essence-of-moessner-s-process}:

\begin{corollary}[Moessner's theorem without dynamic programming]
\label{coro:moessner-s-theorem}
For all $n : \mathbb{N}$ denoted by \texttt{n}
and
for all $x : \mathbb{N}$ denoted by \texttt{x},
evaluating 
\vspace{-1mm}
\begin{center}
\inlinescheme{(moessner-init x n (lambda (x) 1))}
\end{center}
\vspace{-1mm}
yields $\expo{(x + 1)}{n}$.
\end{corollary}

\noindent
In this corollary,
\inlinescheme{(lambda (x) 1)}
implements
$\elam{x}{1}$,
which is enumerated by
$\estreamthree{1}{1}{1}{1}$.

Alternatively, one can also go for another round of summation:

\begin{corollary}[Moessner's theorem with another round of summation]
For all $n : \mathbb{N}$ denoted by \texttt{n}
and
for all $x : \mathbb{N}$ denoted by \texttt{x},
evaluating 
\vspace{-1mm}
\begin{center}
\inlinescheme{(moessner-init x n (lambda (x) (Sigma x (lambda (x) 1))))}
\end{center}
\vspace{-1mm}
yields $\expo{(x + 1)}{n+1}$.
\end{corollary}

\noindent
In this corollary,
\inlinescheme{(lambda (x) (Sigma x (lambda (x) 1)))}
implements
$\elam{x}{\sum_{i=0}^x\,1}$,
which is enumerated by
$\estreamthree{x + 1}{x + 1}{x + 1}{x + 1}$,
for any given $x$.

\subsection{Long's First Theorem}
\label{subsec:long-s-first-theorem}

Long's first theorem~\cite{Long:AMM66} is about Moessner's process when it
starts with a constant stream $\estreamthree{a}{a}{a}{a}$ -- which enumerates $\elam{x}{a}$:

\clearpage

\begin{corollary}[Long's first theorem without dynamic programming]
\label{coro:long-s-first-theorem}
For all $n : \mathbb{N}$ denoted by \texttt{n},
for all $x : \mathbb{N}$ denoted by \texttt{x},
and
for all $a : \mathbb{N}$ denoted by \texttt{a},
evaluating 
\vspace{-1mm}
\begin{center}
\inlinescheme{(moessner-init x n (lambda (x) a))}
\end{center}
\vspace{-1mm}
yields $\expo{a \cdot (x + 1)}{n}$.
\end{corollary}

The phrasing of Moessner's theorem with nested sums
in \sectionref{subsec:moessner-s-process-streamlessly}
makes
Long's first theorem
limpid.
For example, the equality
\begin{center}
$
\elam{x}{a \cdot (x + 1)^3}
 = 
\elam{x}{\sum_{i_1=0}^x\,\sum_{i_2=0}^{\frac{2 \cdot i_1}{1}}\,\sum_{i_3=0}^{\frac{3 \cdot i_2}{2}}\,a}
$
\end{center}

\noindent
holds because of the following factorization: $\forall x : \mathbb{N}$,
\begin{center}
$
\sum_{i_1=0}^x\,\sum_{i_2=0}^{\frac{2 \cdot i_1}{1}}\,\sum_{i_3=0}^{\frac{3 \cdot i_2}{2}}\,a
=
\sum_{i_1=0}^x\,\sum_{i_2=0}^{\frac{2 \cdot i_1}{1}}\,\sum_{i_3=0}^{\frac{3 \cdot i_2}{2}}\,a \cdot 1
=
a \cdot \sum_{i_1=0}^x\,\sum_{i_2=0}^{\frac{2 \cdot i_1}{1}}\,\sum_{i_3=0}^{\frac{3 \cdot i_2}{2}}\,1
$
\end{center}

\subsection{Long's Second Theorem}
\label{subsec:long-s-second-theorem}

Long's second theorem~\cite{Long:AMM66} is about Moessner's process when it
starts with a stream that follows an arithmetic progression:
$\estreamfour{a}{a + d}{a + 2 \cdot d}{a + 3 \cdot d}$.
For this theorem, it is more telling to undo one step of the
process and add one initial iteration, so that instead of starting with
$\estreamthree{1}{1}{1}{1}$,
Moessner's process
starts with $\estreamthree{1}{0}{0}{0}$,
as in the opening of \sectionref{sec:introduction}:
\inputscheme{MOESSNER-INIT-PLUS}

\noindent
Compared to \inlinescheme{moessner-init}, the initial call to \inlinescheme{visit}
is over \inlinescheme{(+ n 1)} instead of \inlinescheme{n}, which provides one more iteration.

\newcommand{\econditionalbroken}[3]{\ensuremath{\mathrm{if}\;{#1}\;\mathrm{then}\;{#2}} \protect\\
\ensuremath{\mathrm{else}\;{#3}}}

Moessner's theorem accounts for the initial stream
$\estreamthree{1}{0}{0}{0}$ -- which enumerates
$\lambda x.\mathrm{if}\;{x = 0}\;\mathrm{then}\;{1}$
$\mathrm{else}\;{0}$:
\begin{corollary}[Moessner without dynamic programming]
For all $n : \mathbb{N}$ denoted by \texttt{n}
and
for all $x : \mathbb{N}$ denoted by \texttt{x},
evaluating 
\vspace{-1mm}
\begin{center}
\inlinescheme{(moessner-init-plus x n (lambda (x) (if (= x 0) 1 0)))}
\end{center}
\vspace{-1mm}
yields $\expo{(x + 1)}{n}$.
\end{corollary}

\noindent
After the new
initial
iteration, the function is enumerated by
$\estreamthree{1}{1}{1}{1}$,
as in Corollary~\ref{coro:moessner-s-theorem}.

Long's first theorem accounts for the initial stream
$\estreamthree{a}{0}{0}{0}$ -- which enumerates
$\lambda x.\mathrm{if}\;{x = 0}\;\mathrm{then}\;{a}$
$\mathrm{else}\;{0}$:
\begin{corollary}[Long's first theorem without dynamic programming]
For all $n : \mathbb{N}$ denoted by \texttt{n},
for all $x : \mathbb{N}$ denoted by \texttt{x},
and
for all $a : \mathbb{N}$ denoted by \texttt{a},
evaluating 
\vspace{-1mm}
\begin{center}
\inlinescheme{(moessner-init-plus x n (lambda (x) (if (= x 0) a 0)))}
\end{center}
\vspace{-1mm}
yields $\expo{a \cdot (x + 1)}{n}$.
\end{corollary}

\noindent
After the new
initial
iteration, the function is enumerated by
$\estreamthree{a}{a}{a}{a}$,
as in Corollary~\ref{coro:long-s-first-theorem}.


Long's second theorem accounts for the initial stream
$\estreamthree{a}{d}{d}{d}$ -- which enumerates
$\lambda x.\mathrm{if}\;{x = 0}$
$\mathrm{then}\;{a}\;\mathrm{else}\;{d}$:

\begin{corollary}[Long's second theorem without dynamic programming]
For all $n : \mathbb{N}$ denoted by \texttt{n},
for all $x : \mathbb{N}$ denoted by \texttt{x},
for all $a : \mathbb{N}$ denoted by \texttt{a},
and
for all $d : \mathbb{N}$ denoted by \texttt{d},
evaluating 
\vspace{-1mm}
\begin{center}
\inlinescheme{(moessner-init-plus x n (lambda (x) (if (= x 0) a d)))}
\end{center}
\vspace{-1mm}
yields $\expo{(a + d \cdot x) \cdot (x + 1)}{n}$.
\end{corollary}

\noindent
After the new initial iteration, the function is enumerated by
$\estreamthree{a}{a + d}{a + 2 \cdot d}{a + 3 \cdot d}$:
Long's arithmetic progression is obtained by this new initial iteration.

\subsection{A Corollary About the Additive Generation of Integral Powers}
\label{subsec:an-alternative-theorem-to-obtain-powers}

For all that the streamless version of Moessner's process presented in \sectionref{subsec:moessner-s-process-streamlessly}
and
the essence of Moessner's theorem presented in \sectionref{sec:the-essence-of-moessner-s-theorem}
involve fewer concepts -- namely only iterated summations and no streams --
they do not shine a new light on Moessner's
insight.
Iterative summations seem intuitive enough,
but where in the spheres did Moessner get the idea of filtering out and
widening the range?

To (unsuccessfully at first, but please do read on) address this question,
let us also parameterize our implementation of Moessner's process with \texttt{g}:
\inputscheme{MOESSNER-FOLD}

\begin{corollary}[Moessner's theorem without dynamic programming]
For all $n : \mathbb{N}$ denoted by \texttt{n}
and
for all $x : \mathbb{N}$ denoted by \texttt{x},
evaluating 
\vspace{-1mm}
\begin{center}
\inlinescheme{(moessner-fold x n (lambda (x) 1) g)}
\end{center}
\vspace{-1mm}
yields $\expo{(x + 1)}{n}$.
\end{corollary}

We are now in position to play with the starting point of the process and with, \eg, its period of elision,
as per Long's program~\cite{Long:FQ86}.
Here is another suggestion, though:
how about defining a version of Moessner's process with a constant filtering-out phase?
Because the result still computes integral powers:

\begin{corollary}[Moessner's theorem without dynamic programming, simpler]
For all $n : \mathbb{N}$ denoted by \texttt{n}
and
for all $x : \mathbb{N}$ denoted by \texttt{x},
evaluating 
\vspace{-1mm}
\begin{center}
\inlinescheme{(moessner-fold x n (lambda (x) 1) (lambda (j _) x))}
\end{center}
\vspace{-1mm}
yields $\expo{(x + 1)}{n}$.
\end{corollary}

Inlining the definition of \texttt{moessner-fold} here gives
rise to the following
stolid
instances of Moessner's theorem for the
exponents 0, 1, 2, 3, and 4 where powers are generated additively, which
was the name of the game:
$$
\begin{array}{@{}l@{\ }c@{\ }l@{}}
\elam{x}{(x + 1)^0}
& = &
\elam{x}{1}
\\[3mm]
\elam{x}{(x + 1)^1}
& = &
\elam{x}{\sum_{i=0}^x\,1}
\\[3mm]
\elam{x}{(x + 1)^2}
& = &
\elam{x}{\sum_{i=0}^x\,\sum_{i=0}^x\,1}
\end{array}
\hspace{2cm}
\begin{array}{@{}l@{\ }c@{\ }l@{}}
\elam{x}{(x + 1)^3}
& = &
\elam{x}{\sum_{i=0}^x\,\sum_{i=0}^x\,\sum_{i=0}^x\,1}
\\[3mm]
\elam{x}{(x + 1)^4}
& = &
\elam{x}{\sum_{i=0}^x\,\sum_{i=0}^x\,\sum_{i=0}^x\,\sum_{i=0}^x\,1}
\\[3mm]
\elam{x}{(x + 1)^5}
& = &
\elam{x}{\sum_{i=0}^x\,\sum_{i=0}^x\,\sum_{i=0}^x\,\sum_{i=0}^x\,\sum_{i=0}^x\,1}
\end{array}
$$

\noindent
This version is immensely simpler to understand since
as foreshadowed in \sectionref{subsubsec:differences-and-antidifferences},
for all expressions $e : \mathbb{N}$,
the equality $\sum_{i=0}^x\,e = (x + 1) \cdot e$ holds
when the local variable $i$ does not occur free in the expression $e$.
It is also immensely simpler to prove
(a routine induction): $\forall x : \mathbb{N},$
\begin{eqnarray*}
1 \;\;=\;\; (x + 1)^0
\;\;\;\;\;\;\;\;\wedge\;\;\;\;\;\;\;\;
\forall n : \mathbb{N},\,
\underbrace{\sum_{i=0}^x\,
            \sum_{i=0}^x\,
            \cdots
            \sum_{i=0}^x\,}_{n+1}\,1
& = &
(x + 1)
\cdot
\underbrace{\sum_{i=0}^x\,
            \cdots
            \sum_{i=0}^x\,}_n\,1
\end{eqnarray*}

A data point: in both versions, experiments consistently show that the
summations induce the \emph{same} number of additions -- namely $(x + 1)^n - 1$,
which is logical without memoization -- to
compute the resulting power if we do not count the increment performed by
$\rawgee$ as an addition.
For example, not only does the following equality hold
\begin{eqnarray*}
\forall x : \mathbb{N}, \;
\sum_{i_1=0}^x\,\sum_{i_2=0}^{\frac{2 \cdot i_1}{1}}\,\sum_{i_3=0}^{\frac{3 \cdot i_2}{2}}\,\sum_{i_4=0}^{\frac{4 \cdot i_3}{3}}\,\sum_{i_5=0}^{\frac{5 \cdot i_4}{4}}\,1
& = &
\sum_{i_1=0}^x\,\sum_{i_2=0}^x\,\sum_{i_3=0}^x\,\sum_{i_4=0}^x\,\sum_{i_5=0}^x\,1
\end{eqnarray*}

\noindent
but
the summations
in the ethereal left-hand side and in the stolid right-hand side
induce the same total number of additions, namely $(x + 1)^5 - 1$,
consistently with Church's $\lambda$ definability exercise of expressing multiplication as iterated addition
and exponentiation as iterated multiplication~\cite{Church:41}.
So Moessner's filtering out with $(j + 1) \cdot x + j$ (in the definition of $\rawgeenop{n}$)
while widening the range of the summation with $\frac{(j + 1) \cdot x}{j}$ (in the definition of $\rawgeen{n}$)
balance out.
As just illustrated though, these two features -- filtering out and
widening the range of the summation -- are unnecessary to compute
integral powers additively.

In passing, the number of additions in the simpler version can be
lowered logarithmically
if multiplication is defined in terms of addition by successive divisions by 2:
\begin{eqnarray*}
\eapp{\eapp{\mathit{times}}{0}}{c}
& = &
0
\\
\eapp{\eapp{\mathit{times}}{(2 \cdot (x + 1))}}{c}
& = &
2 \cdot \eapp{\eapp{\mathit{times}}{(x + 1)}}{c}
\\
\eapp{\eapp{\mathit{times}}{(2 \cdot x + 1)}}{c}
& = &
2 \cdot \eapp{\eapp{\mathit{times}}{x}}{c} + c
\end{eqnarray*}

\noindent
where $2 \cdot \eapp{\eapp{\mathit{times}}{(x + 1)}}{c}$
is not computed as $\eapp{\eapp{\mathit{times}}{(x + 1)}}{c} + \eapp{\eapp{\mathit{times}}{(x + 1)}}{c}$
but with the strict function application $\eapp{\elamp{m}{m + m}}{(\eapp{\eapp{\mathit{times}}{(x + 1)}}{c})}$
and likewise for $2 \cdot \eapp{\eapp{\mathit{times}}{x}}{c}$.
(This strict function application was implemented as a let insertion in Similix~\cite{Bondorf-Danvy:SCP91}.)

All told, %
\begin{eqnarray*}
\forall x : \mathbb{N},\;
(x + 1)^n
& = &
\underbrace{\mathit{times}\,{(x + 1)}\,(
            \mathit{times}\,{(x + 1)}\,(
            \cdots(
            \mathit{times}\,{(x + 1)}}_n\,1)\cdots)).
\end{eqnarray*}

\noindent
This formulation suggests a further improvement
using the computational content of Kleene's $S^m_n$-theorem \cite{Kleene:52},
\ie, partial evaluation~\cite{Consel-Danvy:POPL93}:
(1) specialize \textit{times} with respect to its first argument, and
(2) use this specialized version to carry out the exponentiation.

One can then
re-introduce the dynamic-programming
infrastructure to obtain a stream-based process for generating powers
additively that is simpler
than Moessner's:
$$
\begin{array}{c@{\ }r@{\ }r@{\ }r@{\ }r@{\ \ }l}
0 & 1{\phantom{,}} & 2{\phantom{,}} & 3{\phantom{,}} & \cdots {\phantom{]}} & \textrm{indices in the streams}
\\
\hline
[1, & {\phantom{0}}1, & {\phantom{0}}1, & {\phantom{0}}1, & {\phantom{0}}\cdots[
& \textrm{stream that enumerates } \elam{x}{(x + 1)^0}
\\{}
[1, & {\phantom{0}}2, & {\phantom{0}}3, & {\phantom{0}}4, & \cdots[
& \textrm{stream that enumerates } \elam{x}{(x + 1)^1}
\\{}
[1, & {\phantom{0}}4, & {\phantom{0}}9, & 16, & \cdots[
& \textrm{stream that enumerates } \elam{x}{(x + 1)^2}
\\{}
[1, & {\phantom{0}}8, & 27, & 64, & \cdots[
& \textrm{stream that enumerates } \elam{x}{(x + 1)^3}
\end{array}
$$

\noindent
Given a stream of positive natural numbers exponentiated with $n$,
each number at index $i$ in the stream of positive natural numbers exponentiated with $n + 1$
is obtained by applying $\sum_{x=0}^i$
to the number at index $i$ in the stream of positive natural numbers exponentiated with $n$.
For example, following the first diagonal,
$2 = \sum_{x=0}^1\,1$,
$9 = \sum_{x=0}^2\,3$,
$64 = \sum_{x=0}^3\,16$,
etc.
and each summation can be carried out in logarithmic time rather than in linear time,
which makes this process more efficient than Moessner's.

Also, Long's first theorem still applies: instead of starting from $1$,
one can start from another integer and see the resulting integral power
multiplied by this other integer.

All in all, this alternative additive computation of powers sheds a quantitative light on Moessner's process without memoization
in that it lets us characterize its number of additions, which is the same as the number of additions
when exponentiation is defined by iterated multiplication and multiplication is defined as iterated addition.
But it does not shed a qualitative light on Moessner's idea of
filtering out and widening, since as it turns out, this idea is not needed to compute
integral powers additively.

\subsection{A Corollary About the Additive Generation of Factorial Numbers}
\label{subsec:an-alternative-theorem-to-obtain-factorial-numbers}

The literature contains many variations over Moessner's process that keep
the idea of successively filtering out and computing prefix sums, chiefly
by varying the period of elision in the filtering-out phase, but not only.
Most noticeably, the rules of the game are changed in that the result is
no longer the last stream, but the stream consisting in the first element
of each intermediate stream~\cite{Paasche:MNU53-54}, just like in
Eratosthenes's sieve.
In the literature~\cite{Bickford-al:JLAGMP22,Conway-Guy:TBON96-Moessner-s-Magic,Hinze:IFL08,Hinze:JFP11,Kozen-Silva:AMM13,Long:AMM66,Long:FQ86},
this path is first illustrated with generating Factorial numbers additively,
but we refrain to take it here.
We do note, however, that it is possible to generate Factorial numbers
additively without changing the rules of the game:

\begin{corollary}[Factorial numbers]
\label{coro:factorial-numbers}
For all $n : \mathbb{N}$ denoted by \texttt{n},
evaluating 
\vspace{-1mm}
\begin{center}
\inlinescheme{(moessner-fold 0 n (lambda (x) 1) (lambda (j x) j))}
\end{center}
\vspace{-1mm}
yields $n!$.
\end{corollary}

\noindent
In other words, the function implemented by the Scheme procedure
\vspace{-1mm}
\begin{center}
\inlinescheme{(lambda (n) (moessner-fold 0 n (lambda (x) 1) (lambda (j x) j)))}
\end{center}
\vspace{-1mm}
is enumerated by the stream of Factorial numbers~\cite{OEIS:A000142}.

Again, inlining the definition of \texttt{moessner-fold} here gives rise to
the following
stolid
instances of Moessner's theorem to generate Factorial numbers additively:
\begin{eqnarray*}
1
& = &
0!
\\
\sum_{i=0}^0\,1
& = &
\sum_{i=0}^0\,0!
\;\;=\;\;
1 \cdot 0!
\;\;=\;\;
1!
\\
\sum_{i=0}^1\,\sum_{i=0}^0\,1
& = &
\sum_{i=0}^1\,1!
\;\;=\;\;
2 \cdot 1!
\;\;=\;\;
2!
\\
\sum_{i=0}^2\,\sum_{i=0}^1\,\sum_{i=0}^0\,1
& = &
\sum_{i=0}^2\,2!
\;\;=\;\;
3 \cdot 2!
\;\;=\;\;
3!
\\
\sum_{i=0}^3\,\sum_{i=0}^2\,\sum_{i=0}^1\,\sum_{i=0}^0\,1
& = &
\sum_{i=0}^3\,3!
\;\;=\;\;
4 \cdot 3!
\;\;=\;\;
4!
\\
\sum_{i=0}^4\,\sum_{i=0}^3\,\sum_{i=0}^2\,\sum_{i=0}^1\,\sum_{i=0}^0\,1
& = &
5 \cdot (4 \cdot (3 \cdot (2 \cdot (1 \cdot 1))))
\;\;=\;\;
5!
\end{eqnarray*}

\noindent
This version is simple to understand and simple to prove (another routine
induction):
\begin{eqnarray*}
\sum_{i=0}^0\,1 \;\;=\;\; 1!
\;\;\;\;\;\;\;\;\wedge\;\;\;\;\;\;\;\;
\forall n : \mathbb{N},\;
\underbrace{\sum_{i=0}^{n+1}\,
            \sum_{i=0}^n\,
            \cdots
            \sum_{i=0}^0\,}_{n+2}\,1
& = &
(n + 2)
\cdot
\underbrace{\sum_{i=0}^n\,
            \cdots
            \sum_{i=0}^0\,}_{n+1}\,1
\end{eqnarray*}

We note that the intermediate results are rising Factorial numbers.
These intermediate results could just as well be falling Factorial numbers:

\begin{corollary}[Factorial numbers]
For all $n : \mathbb{N}$ denoted by \texttt{n},
evaluating 
\vspace{-1mm}
\begin{center}
\inlinescheme{(moessner-fold 0 n (lambda (x) 1) (lambda (j x) (- n j)))}
\end{center}
\vspace{-1mm}
yields $n!$.
\end{corollary}

Indeed, inlining the definition of \texttt{moessner-fold} here gives rise
to the following instance of Moessner's theorem to generate Factorial
numbers additively:
\vspace{-2mm}~\begin{eqnarray*}
\sum_{i=0}^0\,\sum_{i=0}^1\,\sum_{i=0}^2\,\sum_{i=0}^3\,\sum_{i=0}^4\,1
& = &
1 \cdot (2 \cdot (3 \cdot (4 \cdot (5 \cdot 1))))
\;\;=\;\;
5!
\end{eqnarray*}

\noindent
In fact, though, any permutation does the job since multiplication is left-permutative~\cite{Danvy:JFP23}:
\begin{eqnarray*}
\sum_{i=0}^3\,\sum_{i=0}^1\,\sum_{i=0}^4\,\sum_{i=0}^0\,\sum_{i=0}^2\,1
& = &
4 \cdot (2 \cdot (5 \cdot (1 \cdot (3 \cdot 1))))
\;\;=\;\;
5!
\end{eqnarray*}

And Long's first theorem still applies: instead of starting from $1$,
one can start from another integer and see the resulting Factorial number
multiplied by this other integer.
Likewise for generalizing the first argument of \inlinescheme{moessner-fold}
from
$0$
to a non-negative integer
$x$
in Corollary~\ref{coro:factorial-numbers}:

\begin{corollary}[Multiples of factorial numbers]
For all $n : \mathbb{N}$ denoted by \texttt{n}
and
for all $x : \mathbb{N}$ denoted by \texttt{x},
evaluating 
\vspace{-1mm}
\begin{center}
\inlinescheme{(moessner-fold x n (lambda (x) 1) (lambda (j x) j))}
\end{center}
\vspace{-1mm}
yields $n! \cdot (x + 1)$.
\end{corollary}

\noindent
(Rationale: The stolid instance reads
$\sum_{i=0}^n\,
 \cdots
 \sum_{i=0}^1\,
 \sum_{i=0}^x\,
 1$
instead of
$\sum_{i=0}^n\,
 \cdots
 \sum_{i=0}^1\,
 \sum_{i=0}^0\,
 1$.)

\subsection{A Generalization: Primitive Iteration}

Factorial numbers and integral powers are both specified multiplicatively.
So, generalizing, $\forall f : \mathbb{N} \rightarrow \mathbb{N},$
\begin{eqnarray*}
\prod_{i=0}^n\, \eapp{f}{i}
\;\;=\;\;
\eapp{f}{0} \cdot \eapp{f}{1} \cdot \ldots \cdot \eapp{f}{n}
& = &
\eapp{f}{0} \cdot (\eapp{f}{1} \cdot (\ldots \cdot (\eapp{f}{n} \cdot 1)\ldots))
\;\;=\;\;
\sum_{i=0}^{\eappp{f}{0} - 1}
\ldots
\sum_{i=0}^{\eappp{f}{n} - 1}\,
1
\end{eqnarray*}

\begin{corollary}[Products]
For all $f : \mathbb{N} \rightarrow \mathbb{N}$ denoted by \texttt{f}
and
for all $n : \mathbb{N}$ denoted by \texttt{n},
evaluating 
\vspace{-1mm}
\begin{center}
\inlinescheme{(moessner-fold 0 n (lambda (x) 1) (lambda (j x) (- (f j) 1)))}
\end{center}
\vspace{-1mm}
yields $\prod_{i=0}^n\,\eapp{f}{i}$.
\end{corollary}

\noindent
This corollary suggests a variant of \inlinescheme{moessner-fold} where
summation starts at 1 instead of at 0, since

\begin{eqnarray*}
\sum_{i=0}^{\eappp{f}{0} - 1}
\ldots
\sum_{i=0}^{\eappp{f}{n} - 1}\,
1
& = &
\sum_{i=1}^{\eapp{f}{0}}
\ldots
\sum_{i=1}^{\eapp{f}{n}}\,
1
\end{eqnarray*}

\noindent
But be that as it may, we can obtain
multiplicative
numbers without changing the name of the game:

\begin{corollary}[$x$-fold factorial numbers]
For all $n : \mathbb{N}$ denoted by \texttt{n}
and
for all $x : \mathbb{N}$ denoted by \texttt{x},
evaluating 
\vspace{-1mm}
\begin{center}
\inlinescheme{(moessner-fold x n (lambda (x) 1) (lambda (j _) (* (+ j 1) x)))}
\end{center}
\vspace{-1mm}
yields
$1$ when $x=0$,
the $n + 1$st factorial number when $x=1$,
a double factorial number of odd numbers when $x=2$,
the $n + 1$st triple factorial number when $x=3$,
the $n + 1$st quartic (or 4-fold) factorial number when $x=4$,
the $n + 1$st quintuple factorial number when $x=5$,
etc., according to the On-Line Encyclopedia of Integer Sequences.
\end{corollary}

Inlining the definition of \texttt{moessner-fold} here gives rise to the
following stolid instance of Moessner's theorem:
$$
\sum_{i=0}^{1 \cdot x}\,
\sum_{i=0}^{2 \cdot x}\,
\sum_{i=0}^{3 \cdot x}\,
\cdots
\sum_{i=0}^{(n + 1) \cdot x}\,
1
$$

\noindent
(The multiplications in the upper bounds can be obtained by iterated
addition, an example of third-order primitive iteration -- one over $n$,
one over each upper bound, and one to carry out the multiplication.
Ditto for superfactorial numbers.)

\subsection{A Corollary About the Additive Generation of Binomial Coefficients}
\label{subsec:a-corollary-to-obtain-binomial-coefficients}

\begin{corollary}[Binomial coefficients]
For all $x : \mathbb{N}$ denoted by \texttt{x}
and
for all $n : \mathbb{N}$ denoted by \texttt{n},
evaluating 
\vspace{-1mm}
\begin{center}
\inlinescheme{(moessner-fold x n (lambda (x) 1) (lambda (j x) x))}
\end{center}
\vspace{-1mm}
which is Moessner's process without striking-out phase,
yields $\binom{x + n}{n}$.
\end{corollary}

\noindent
In other words, and remembering Pascal's triangle
$$
\begin{array}{@{}c@{\ }c@{\ }c@{\ }c@{\ }c@{\ }c@{\ }@{\ }c@{}}
&&&\binom{0}{0}&&&
\\
&&\binom{1}{0}&&\binom{1}{1}&&
\\
&\binom{2}{0}&&\binom{2}{1}&&\binom{2}{2}&
\\
\binom{3}{0}&&\binom{3}{1}&&\binom{3}{2}&&\binom{3}{3}
\end{array}
$$
\begin{itemize}[leftmargin=3.5mm]

\item
  for all $x : \mathbb{N}$ denoted by \texttt{x},
  the function implemented by 
  \vspace{-1mm}
  \begin{center}
  \inlinescheme{(lambda (n) (moessner-fold x n (lambda (x) 1) (lambda (j x) x)))}
  \end{center}
  \vspace{-1mm}
  is enumerated by the $x + 1$st diagonal of Pascal's triangle from right to left
  (reminder: $0$ is the first Peano number, $1$ is the second, etc.), and

\item
  for all $n : \mathbb{N}$ denoted by \texttt{n},
  the function implemented by 
  \vspace{-1.5mm}
  \begin{center}
  \inlinescheme{(lambda (x) (moessner-fold x n (lambda (x) 1) (lambda (j x) x)))}
  \end{center}
  \vspace{-1.5mm}
  is enumerated by the $n + 1$st diagonal of Pascal's triangle from left to right
  (reminder: $0$ is the first Peano number, $1$ is the second, etc.).

\end{itemize}

\vspace{-1mm}

Inlining the definition of \texttt{moessner-fold} here gives rise to
the following
ethereal
instances of Moessner's theorem to generate binomial coefficients additively
(see \corollaryref{coro:binomial-coefficients} in \sectionref{sec:the-essence-of-moessner-s-theorem}):
$$
\begin{array}{@{}l@{\ }c@{\ }l@{}}
\binom{x+0}{0}
& = &
1
\\[3mm]
\binom{x+1}{1}
& = &
\sum_{i_1=0}^x\,1
\\[3mm]
\binom{x+2}{2}
& = &
\sum_{i_1=0}^x\,\sum_{i_2=0}^{i_1}\,1
\end{array}
\hspace{2cm}
\begin{array}{@{}l@{\ }c@{\ }l@{}}
\binom{x+3}{3}
& = &
\sum_{i_1=0}^x\,\sum_{i_2=0}^{i_1}\,\sum_{i_3=0}^{i_2}\,1
\\[3mm]
\binom{x+4}{4}
& = &
\sum_{i_1=0}^x\,\sum_{i_2=0}^{i_1}\,\sum_{i_3=0}^{i_2}\,\sum_{i_4=0}^{i_3}\,1
\\[3mm]
\binom{x+5}{5}
& = &
\sum_{i_1=0}^x\,\sum_{i_2=0}^{i_1}\,\sum_{i_3=0}^{i_2}\,\sum_{i_4=0}^{i_3}\,\sum_{i_5=0}^{i_4}\,1
\end{array}
$$

\subsection{A Corollary About the Additive Generation of Catalan Numbers}
\label{subsec:a-corollary-to-obtain-catalan-numbers}

In the Online Encyclopedia of Integer Sequences~\cite{OEIS}, the page for
Catalan numbers~\cite{OEIS:CN} is presented as the longest of the whole
encyclopedia, ``for good reasons'' since they arise in so many situations
-- for example, trees, as in Hinze's work~\cite[Section~5.5]{Hinze:JFP11},
which
appears to be 
the only mention of Catalan numbers in the literature about Moessner's theorem.
Here is their inductive definition and a definition in closed form:
\[
\left\{
\begin{array}{rcl}
C_0 & = & 1
\\
\forall n : \mathbb{N},\;
C_{n+1}
& = &
\frac{\displaystyle 2 \cdot (2 \cdot n + 1) \cdot C_n}
     {\displaystyle n+2}
\end{array}
\right.
\phantom{\hspace{3cm}}
\forall n : \mathbb{N},\;
C_n
=
\frac{\displaystyle \binom{2 \cdot n}{n}}
     {\displaystyle n + 1}
\]

\begin{corollary}[Catalan numbers]
\label{coro:catalan-numbers}
For all $n : \mathbb{N}$ denoted by \texttt{n},
evaluating 
\vspace{-1mm}
\begin{center}
\inlinescheme{(moessner-fold 0 n (lambda (x) 1) (lambda (j x) (+ x 1)))}
\end{center}
\vspace{-1mm}
yields the $n + 1$st Catalan number.
\end{corollary}

\noindent
In other words, the function implemented by the Scheme procedure
\vspace{-1mm}
\begin{center}
\inlinescheme{(lambda (n) (moessner-fold 0 n (lambda (x) 1) (lambda (j x) (+ x 1))))}
\end{center}
\vspace{-1mm}
is enumerated by the stream of Catalan numbers.

Inlining the definition of \texttt{moessner-fold} here gives rise to
the following
ethereal
instances of Moessner's theorem to generate Catalan numbers additively:
$$
\begin{array}{@{}l@{\ }c@{\ }l@{}}
C_0
& = &
1
\\[3mm]
C_1
& = &
\sum_{i_1=0}^0\,1
\; = \; 1
\\[3mm]
C_2
& = &
\sum_{i_1=0}^0\,\sum_{i_2=0}^{i_1+1}\,1
\; = \; 2
\end{array}
\hspace{2cm}
\begin{array}{@{}l@{\ }c@{\ }l@{}}
C_3
& = &
\sum_{i_1=0}^0\,\sum_{i_2=0}^{i_1+1}\,\sum_{i_3=0}^{i_2+1}\,1
\; = \; 5
\\[3mm]
C_4
& = &
\sum_{i_1=0}^0\,\sum_{i_2=0}^{i_1+1}\,\sum_{i_3=0}^{i_2+1}\,\sum_{i_4=0}^{i_3+1}\,1
\; = \; 14
\\[3mm]
C_5
& = &
\sum_{i_1=0}^0\,\sum_{i_2=0}^{i_1+1}\,\sum_{i_3=0}^{i_2+1}\,\sum_{i_4=0}^{i_3+1}\,\sum_{i_5=0}^{i_4+1}\,1
\; = \; 42
\end{array}
$$

\noindent
At the time of writing,
these nested summations are not mentioned in the On-Line Encyclopedia of Integer Sequences.
They do, however,
correspond to
the nested for loops in Example~27.2, page 222,
of Ralph Grimaldi's introduction to Fibonacci and Catalan numbers~\cite{Grimaldi:12},
as detailed in \sectionref{subsec:catalan-numbers}.

Long's first theorem applies again: instead of starting from $1$,
one can start from another integer and see the resulting Catalan number
multiplied by this other integer.

In Corollary~\ref{coro:catalan-numbers},
and as mentioned in \sectionref{sec:the-essence-of-moessner-s-theorem},
generalizing the first argument of \inlinescheme{moessner-fold}
from
$0$
to a non-negative integer
$x$
gives rise to the $x + 1$st convolution of Catalan numbers.
Generalizing
the fourth argument of \inlinescheme{moessner-fold}
from
\inlinescheme{(lambda (j x) (+ x 1))}
to
\inlinescheme{(lambda (j x) (+ x i))}
where \inlinescheme{i} denotes a
positive
integer,
however,
does not give rise to a discernible pattern.
(Amusingly, though, the consecutive integers 3, 4, 5, and 6 give rise
to the consecutive OEIS integer sequences A002293, A002294, A002295, and A002296.)
As mentioned in \sectionref{sec:the-essence-of-moessner-s-theorem},
letting \inlinescheme{i} be \inlinescheme{0} gives rise to binomial coefficients,
and letting \inlinescheme{i} be \inlinescheme{(quotient x j)} gives rise to integral powers.

\subsection{Moessner's Magic: Primitive Recursion}

Could one express Catalan numbers as a stolid instance of Moessner's
theorem rather than as the ethereal instance in
\sectionref{subsec:a-corollary-to-obtain-catalan-numbers}?
This looks unlikely because of the inductive specification of Catalan numbers
(stated in \sectionref{subsec:a-corollary-to-obtain-catalan-numbers}).
The issue is that ${2 \cdot (2 \cdot n + 1) \cdot C_n}$ is divisible by ${n + 2}$,
but ${2 \cdot (2 \cdot n + 1)}$ is not.
And so even though Catalan numbers are integers, the equality
\begin{eqnarray*}
\forall n : \mathbb{N},\;
\frac{2 \cdot (2 \cdot n + 1) \cdot C_n}
     {n + 2}
& = &
\frac{2 \cdot (2 \cdot n + 1)}
     {n + 2}
\cdot
C_n
\end{eqnarray*}

\noindent
is only valid for rational numbers, not for integer division.
And since the upper bound of a summation is an integer, not a rational
number, this recurrence is not multiplicative and thus not expressible as
a stolid instance of Moessner's theorem where the upper bounds are
independent of the outer indices.
Ditto for the multiplicative definition of Catalan numbers:
\begin{eqnarray*}
\forall n : \mathbb{N},\;
C_n
& = &
\prod_{i=2}^n\,\frac{n + i}{i}
\end{eqnarray*}

And therein lies the magic of Moessner's theorem: the upper bounds of the
inner sums depend on the indices of the outer summations, and these
dependencies account for the transitory rational arithmetic using integer
arithmetic in the ethereal instance of Moessner's theorem to express
Catalan numbers.

The executive summary and
\sectionref{sec:the-essence-of-moessner-s-theorem} presented a series of
examples that illustrate
these
dependencies.

\section{Back to Dynamic Programming}
\label{sec:back-to-dynamic-programming}

For completeness, let us revisit the additive generation of powers and
illustrate how to express it using dynamic programming.

Consider, for example, the Scheme procedure that maps $x$ to $(x + 1)^5$:
\inputscheme{POWER-5}

\noindent
Computing $(x + 1)^5$ with \inlinescheme{power-5} is achieved by performing $(x + 1)^5$ additions.
(In general, for any $n : \mathbb{N}$,
computing $(x + 1)^n$ is achieved by performing $(x + 1)^n$ additions.)

Since the local procedures \inlinescheme{p5}, \inlinescheme{p4},
\inlinescheme{p3}, and \inlinescheme{p2} are repeatedly called with the same arguments,
let us uniformly cache their results into lists,
which is the essence of dynamic programming:
\inputscheme{MPOWER-5_A}
\inputscheme{MPOWER-5_B}

\noindent
where applying \inlinescheme{iotap} to a non-negative integer $n$ constructs an increasing list from $0$ to $n$.
Computing $(x + 1)^5$ with \inlinescheme{mpower-5} is achieved by performing $\frac{5 \cdot 6}{2} \cdot x$ additions.
(In general, for any $n : \mathbb{N}$ that is not $2$ nor $3$,
computing $(x + 1)^n$ is achieved by performing $x \cdot \sum_{i=0}^n\,i$ additions
with dynamic programming -- which illustrates the impact of dynamic programming here:
from $(x + 1)^n$ to $x \cdot \sum_{i=0}^n\,i$ additions.)

We can also fuse the combination of \inlinescheme{map},
\inlinescheme{Sigma}, and \inlinescheme{iotap} into a procedure that
is listless:
\inputscheme{MAP-SIGMA-IOTAP}

\noindent
(Such a ``loop fusion'' is a classical program transformation that dates back to Burge and Landin~\cite{Burge:75},
and a ``listless'' program is one that does not create intermediate lists~\cite{Wadler:LFP84}.)
This procedure is listless because the only list it constructs is the result.
The corresponding power function reads as follows:
\inputscheme{MPOWER-5_LISTLESS}

We can also tune the listless combination
so that it does not construct the elements of the lists that will be skipped in the next place:
\inputscheme{MAP-SIGMA-IOTAP_SKIP}

\noindent
The skipping part is achieved with the inner conditional expression.

The corresponding power function then does not need to skip elements in the lists,
making it clearer that $x \cdot \sum_{i=0}^5\,i$ additions are performed:
\inputscheme{MPOWER-5_LISTLESS_}

\noindent
Replacing \inlinescheme{Sigma x} with
\inlinescheme{fused_ x 1}
makes \inlinescheme{mpower-5_listless_} map $x$ to the increasing list
$[1^5,\,2^5,\cdots,$ $(x + 1)^5]$
instead of to $(x + 1)^5$,
performing the same number of additions thanks to the caching of dynamic programming.

Overall, here is Moessner's process with dynamic programming
where the striking-out phase is anticipated:
\inputscheme{MOESSNER-WITH-DYNAMIC-PROGRAMMING}

\noindent
In practice, we
use
arrays rather than lists since we know their size
and since unlike lists, they are indexed in constant time, not in linear time.

\section{Related Work About Nested Sums}
\label{sec:related-work-about-nested-sums}

With the exception of Irwin and Lahlou's work described in
\sectionref{subsec:number-of-types-of-binary-trees-of-a-given-height}
and of Baumann's ``$k$-dimensionale Champagnerpyramide''~\cite{Baumann:MS18},
the author could not find much about nested sums in the literature.

\subsection{Sums and For Loops}

In his lecture notes on Discrete Mathematics~\cite[Chapter~6]{Aspnes:22}, Aspnes
connects sums and for loops as an evidence.
(To quote, ``Mathematicians invented this notation centuries ago because they didn't
have for loops.'')
So for two imperative examples, %

\begin{itemize}

\item
the OCaml expression
\begin{small}
\vspace{-2mm}
\begin{verbatim}
  let a = ref 0
  in for i = 0 to 5 do
       a := !a + 1
     done;
     !a
\end{verbatim}
\vspace{-2mm}
\end{small}
evaluates to $\sum_{i=0}^5\,1$, \ie, to $6$, and %

\item
the OCaml expression
\begin{small}
\vspace{-2mm}
\begin{verbatim}
  let a = ref 0
  in for i = 1 to 5 do
       a := !a + 1
     done;
     !a
\end{verbatim}
\vspace{-2mm}
\end{small}
evaluates to $\sum_{i=1}^5\,1$, \ie, to $5$.
\end{itemize}

Regarding nested sums, Aspnes mentions that
$\sum_{i=1}^a\,\sum_{j=1}^b\,1$ is a ``rather painful'' way to multiply
$a$ and $b$, and Example~9, page~408 of the 8th edition of Rosen's
textbook about discrete mathematics and its
applications~\cite[Section~6.1.1]{Rosen:19}, illustrates nested for loops
to compute the product of $n$ numbers.
To wit, the OCaml expression
{\small\texttt{example\_Rosen 2 3 4}}
evaluates to $\sum_{i_1=1}^2\,\sum_{i_2=1}^3\,\sum_{i_3=1}^4\,1$,
\ie, to $2 \cdot 3 \cdot 4$, \ie, to $24$,
given the following OCaml declaration: %

\begin{small}
\begin{verbatim}
  let example_Rosen n1 n2 n3 =
    let k = ref 0
    in for i1 = 1 to n1 do
         for i2 = 1 to n2 do
           for i3 = 1 to n3 do
             k := !k + 1
           done
         done
       done;
       !k
\end{verbatim}
\end{small}

In Section~6.1.8 his lecture notes~\cite{Aspnes:22},
Aspnes also gives
$\sum_{i=0}^n\,\sum_{j=0}^i\,(i + 1) \cdot (j + 1)$ as an example
of a nested sum where the upper bound of the inner sum
is the index of the outer sum, and mentions that ``[f]or larger $n$, the
number of [products] grows quickly.''

\subsection{Primitive iteration and primitive recursion}

A for loop whose body uses its index is said to be primitive recursive
and
a for loop whose body does not use its index is said to be primitive iterative~\cite{Danvy:JFP23}.

\subsection{Multiset Coefficients}
\label{subsec:multiset-coefficients}

In Example~6, page~449 of the aforementioned textbook about discrete
mathematics and its applications~\cite[Section~6.5.3]{Rosen:19},
Rosen states the nested for loops that correspond to
$$
 \sum_{i_1=1}^x\,
 \sum_{i_2=1}^{i_1}\,
 \cdots\,
 \sum_{i_n=1}^{i_{n-1}}\,
 1
$$
for two given natural numbers $x$ and $n$ (renamed here for notational
consistency).
He proves
combinatorially
that the result is
the number of combinations of $n$ elements from a multiset with $x$ elements,
\ie, from a set where elements can occur several times:
the result is the multiset coefficient $\mulset{x}{n}$.

\begin{theorem}[Multiset coefficients as nested sums (Rosen)]
$
 \forall x \in \mathbb{N},
 \forall n \in \mathbb{N},
$
\begin{eqnarray*}
  \sum_{i_1=1}^x\,
  \sum_{i_2=1}^{i_1}\,
  \cdots\,
  \sum_{i_n=1}^{i_{n-1}}\,
  1
  & = &
  \mmulset{x}{n}
  \;\;=\;\;
  \binom{x + n - 1}{n}
\end{eqnarray*}
\end{theorem}

\begin{proof}[{\rm\textbf{Proof}}] %
By nested first-order induction on $n$ and then on $x$.
\end{proof}

For comparison, \corollaryref{coro:binomial-coefficients} is proved by
induction on $n$, using \lemmaref{lem:summing-binomial-coefficients} just
below.
The theorem and its proof rely
on the following series of
equalities:
\begin{description}

\item[$n = 0$]
  $\forall x \in \mathbb{N},\,
   1
   =
   \binom{x + 0}{0}$

\item[$n = 1$]
  $\forall x \in \mathbb{N},\,
   \sum_{i_1=0}^x\,
   1
   =
   \sum_{i_1=0}^x\,
   \binom{i_1 + 0}{0}
   =
   \binom{x + 1}{1}$

\item[$n = 2$]
  $\forall x \in \mathbb{N},\,
   \sum_{i_2=0}^x\,
   \sum_{i_1=0}^{i_2}\,
   1
   =
   \sum_{i_2=0}^x\,
   \binom{i_2 + 1}{1}
   =
   \binom{x + 2}{2}$

\item[$n = 3$]
  $\forall x \in \mathbb{N},\,
   \sum_{i_3=0}^x\,
   \sum_{i_2=0}^{i_3}\,
   \sum_{i_1=0}^{i_2}\,
   1
   =
   \sum_{i_3=0}^x\,
   \binom{i_3 + 2}{2}
   =
   \binom{x + 3}{3}$

\end{description}
etc.

\begin{lemma}
\label{lem:summing-binomial-coefficients}
$\forall x \in \mathbb{N},
 \forall n \in \mathbb{N},
 \sum_{i=0}^x\,\binom{i + n}{n}
 =
 \binom{x + n + 1}{n + 1}$
\end{lemma}

\begin{proof}[{\rm\textbf{Proof}}] %
By induction on $x$.
\end{proof}

\begin{corollary}
\label{cor:multiset-coefficients-from-0-and-from-1}
$
 \forall x \in \mathbb{N},
 \forall n \in \mathbb{N},
  \sum_{i_1=0}^{x}\,
  \sum_{i_2=0}^{i_1}\,
  \cdots\,
  \sum_{i_n=0}^{i_{n-1}}\,
  1
  \;=\;
  \sum_{i_1=1}^{x+1}\,
  \sum_{i_2=1}^{i_1}\,
  \cdots\,
  \sum_{i_n=1}^{i_{n-1}}\,
  1
$
\end{corollary}

\subsection{Catalan Numbers}
\label{subsec:catalan-numbers}

Catalan numbers (A000108
$= \estreamnine{1}{1}{2}{5}{14}{42}{132}{429}{1430}$)
can also be obtained using nested for loops
(Example~27.2, page~222, of Grimaldi's Introduction to Fibonacci and
Catalan numbers~\cite[Chapter~27]{Grimaldi:12} and Example~9, page~150, of the
second edition of the Handbook of Discrete and Combinatorial
Mathematics~\cite[Section~3.1.3]{Rosen:18}) and via a version of
Moessner's process without dynamic programming:
\begin{eqnarray*}
  \forall n \in \mathbb{N},
  \sum_{i_1=1}^1\,
  \sum_{i_2=1}^{i_1 + 1}\,
  \sum_{i_3=1}^{i_2 + 1}\,
  \cdots
  \sum_{i_n=1}^{i_{n-1} + 1}\,
  1
  & = &
  C_{n}
\\
  \forall n \in \mathbb{N},
  \sum_{i_0=0}^0\,
  \sum_{i_1=0}^{i_0 + 1}\,
  \sum_{i_2=0}^{i_1 + 1}\,
  \cdots
  \sum_{i_n=0}^{i_{n-1} + 1}\,
  1
  & = &
  C_{n+1}
\end{eqnarray*}

Grimaldi sketches how the growing indices match the structure of a tree
where subtrees have a growing number of subtrees (Figure~27.3, page~222),
which justifies how these nested sums compute Catalan numbers,
since the number of nodes at level $n$ of this tree is $C_{n}$.

In the second equation, and according to the On-Line Encyclopedia of
Integer Sequences, generalizing $0$ to the natural number $x$ appears to
give rise to the $x + 1$st convolution of Catalan numbers:
A000245 when $x$ is $2$,
A002057 when $x$ is $3$,
A000344 when $x$ is $4$,
A003517 when $x$ is $5$,
etc.

\subsection{Number of Types of Binary Trees of a Given Height}
\label{subsec:number-of-types-of-binary-trees-of-a-given-height}

A variant of Moessner's process without dynamic programming gives rise to
the following nested sums that compute the number of
types of
binary trees of height $n$, \ie, to A002449
$= \estreamsix{1}{1}{2}{6}{26}{166}$:
\begin{eqnarray*}
  \forall n \in \mathbb{N}, \;
  \sum_{i_0=0}^1\,
  \sum_{i_1=0}^{2 \cdot i_0 + 1}\,
  \sum_{i_2=0}^{2 \cdot i_1 + 1}\,
  \cdots
  \sum_{i_n=0}^{2 \cdot i_{n-1} + 1}\,
  1
  & = &
  \mathrm{A002449}_{n+2}
\end{eqnarray*}

\noindent
(Replacing the opening $1$ with $2$ and the $2$ factor with $3$ gives rise to ternary trees,
and likewise for quaternary trees, etc.)

In the web page dedicated to A002449~\cite{OEIS:A002449}, Irwin conjectures that the following
nested sums also compute A002449:
\begin{eqnarray*}
  \mathrm{A002449}_{n+2}
  & = &
  \sum_{i_1=1}^2\,
  \sum_{i_2=1}^{2 \cdot i_1}\,
  \cdots\,
  \sum_{i_{n-1}=1}^{2 \cdot i_{n-2}}\,
  \sum_{i_{n}=1}^{2 \cdot i_{n-1}}\,
  2 \cdot i_n,
  \; \mathrm{for} \: n \geq 1
\end{eqnarray*}

\noindent
In the material~\cite{Lahlou:A002449} that accompanies A002449~\cite{OEIS:A002449}, Lahlou proves this
conjecture using West's generating trees~\cite{West:DM96},
which is essentially the same coinductive argument as Grimaldi's for Catalan numbers
(see \sectionref{subsec:catalan-numbers}).
Lahlou also lists a dozen nested summations that compute known integer
sequences, including Catalan numbers and Fibonacci numbers:
\begin{theorem}[Fibonacci numbers as nested sums (Lahlou)]
$\forall n \in \mathbb{N},
 A_n = F_{n+1}$
where
$$
  A_n
  \;=\;
  \left\{
  \begin{array}{ll}
    1
    &
    \mathit{if} \; n = 0
    \\%
    1
    &
    \mathit{if} \; n = 1
    \\
    \sum_{i_2=1}^1\,
    \sum_{i_3=1}^{3 - i_2}\,
    \cdots\,
    \sum_{i_n=1}^{3 - i_{n-1}}\,
    (3 - i_n)
    &
    \mathit{otherwise}
  \end{array}
  \right.
$$
\end{theorem}

\begin{proof}[{\rm\textbf{Proof}}] %
By induction on $n$.
\end{proof}

Irwin's conjecture, West's generating trees, and Lahlou's note are the closest related work here.

\section*{Acknowledgments}
The author is grateful to Julia Lawall for tirelessly providing grounded comments
throughout the writing of this article.
Thanks are also due
to Kira Kutscher for lending an informed hand with the early articles
about Moessner's theorem as well as for pertinent comments,
to the NUS library for providing the complete collection of the issues of Sphinx and of Scripta Mathematica,
to its librarians for their diligent competence,
to Kira Kutscher and Lukas Fesser for their help elucidating Alfred Moessner's
particulars,
to the anonymous reviewers for their time and comments,
and
to Annette Bieniusa, Markus Degen, and Stefan Wehr for their editorship.

This article is dedicated to Peter Thiemann
in friendly homage for his lifetime of scholarship, academic dedication, and scientific contributions.

\clearpage

\bibliography{../../mybib.bib}
\bibliographystyle{eptcs}

\appendix

\section{Content of the accompanying .scm file}

The accompanying \texttt{.scm} file contains an implementation in Scheme of the whole article
together with a comprehensive collection of tests.


%

\section{Content of the accompanying .v file}

The accompanying \texttt{.v} file contains a formalization in Coq of part of
the executive summary,
\sectionref{sec:the-essence-of-moessner-s-theorem},
\sectionref{sec:the-essence-of-moessner-s-process},
and \sectionref{sec:related-work-about-nested-sums}.


%

%
%
%
%
%
%
%
%
%
%
%
%
%
%
%
%
%
%
%
%
%
%
%
%
%
%
%
%
%

%

\section{A Curiosa About Polygonal Numbers}
\label{app:a-curiosa-about-polygonal-numbers}


\subsection{Introduction}

Polygonal numbers are classical, textbook material about the positive
integers whose successor constructors in base 1 (Peano numbers) can be
laid down on the plane in a way that represents a polygon:
a triangle (the numbers are triangular: 1, 3, 6, 10, 15, \ldots),
a square (the numbers are square: 1, 4, 9, 16, 25, \ldots),
a pentagon (the numbers are pentagonal: 1, 5, 12, 22, 35, \ldots),
etc.
They
are known since the Greeks~\cite{Dixon:20} and are a
thoroughly explored topic with many beautiful
properties~\cite{Conway-Guy:TBON96-the-polygonal-numbers, Deza-Deza:12}.
For all positive natural numbers $n$,
$\frac{n \cdot (n + 1)}{2}$ denotes a triangular number,
$n^2$ denotes a square number,
$\frac{n \cdot (3 \cdot n - 1)}{2}$ denotes a pentagonal number,
$2 \cdot n^2 - n$ denotes an hexagonal number,
and more generally,
$\sum_{i=0}^n (k \cdot i + 1)$ denotes a $k + 2$nd-order polygonal number,
which can also be expressed using triangular numbers of order $n$,
for all positive natural numbers $k$.


\subsection{The Curiosa}

In the course of \sectionref{subsec:moessner-s-process-streamlessly},
page~\pageref{page:the-origin-of-the-curiosa-about-polygonal-numbers},
the following sum pops up (where $n$ is $i_1$, $i$ is $i_2$, and $i_1$ is
quantified in a sum, not with $\lambda$):
$$
   \elam{n}{\sum_{i=0}^{2 \cdot n}\, \frac{i}{2}}
   :
   \mathbb{N} \rightarrow \mathbb{N}
$$

\noindent
As it happens, this function is enumerated by the stream of square numbers,
\ie, $\estreamfive{0}{1}{4}{9}{16}$.
So it is equivalent to $\elam{n}{n^2}$.
In words, summing twice as many successive halves yields a square number:

\begin{theorem}[Square numbers as bounded summations of integer quotients] \ \\
$\forall n \in \mathbb{N},
 \sum_{i=0}^{2 \cdot n} \, \equofloor{i}{2}
 =
 n^2$
\end{theorem}

\begin{proof}
By induction on $n$.

\noindent
Base case ($n = 0$):
We want to show $\sum_{i=0}^{2 \cdot 0} \, \equofloor{i}{2} = 0^2$.

\noindent
$\sum_{i=0}^{2 \cdot 0} \, \equofloor{i}{2}
 =
 \sum_{i=0}^0 \, \equofloor{i}{2}
 =
 \equofloor{0}{2}
 =
 0
 =
 0^2$

\noindent
Induction step ($n = n' + 1$):
Assuming $\sum_{i=0}^{2 \cdot n'} \, \equofloor{i}{2} = n'^2$,
we want to show $\sum_{i=0}^{2 \cdot (n' + 1)} \, \equofloor{i}{2} = (n' + 1)^2$.

\noindent
$\sum_{i=0}^{2 \cdot (n' + 1)} \, \equofloor{i}{2}
 =
 \sum_{i=0}^{2 \cdot n' + 2} \, \equofloor{i}{2}
 =
 \sum_{i=0}^{2 \cdot n'} \, \equofloor{i}{2} + \sum_{i=2 \cdot n' + 1}^{2 \cdot n' + 2} \, \equofloor{i}{2}
 =
 n'^2 + \equofloor{2 \cdot n' + 1}{2} + \equofloor{2 \cdot n' + 2}{2}
 =
 n'^2 + n' + (n' + 1)
 =
 n'^2 + 2 \cdot n' + 1
 =
 (n' + 1)^2$
\end{proof}

Curiosity compels one to look at
$$
   \elam{n}{\sum_{i=0}^{3 \cdot n}\, \frac{i}{3}}
   :
   \mathbb{N} \rightarrow \mathbb{N}
$$

\noindent
which -- via the On-Line Encyclopedia of Integer Sequences -- is
enumerated by the stream of pentagonal numbers~\cite{OEIS:A000326},
\ie, $\estreamsix{0}{1}{5}{12}{22}{35}$.
So it is equivalent to $\elam{n}{\frac{n \cdot (3 \cdot n - 1)}{2}}$.
In words, summing three times as many successive thirds yields a pentagonal number:

\begin{theorem}[Pentagonal numbers as bounded summations of integer quotients] \ \\
$\forall n \in \mathbb{N},
 \sum_{i=0}^{3 \cdot n} \, \equofloor{i}{3}
 =
 \frac{n \cdot (3 \cdot n - 1)}{2}$
\end{theorem}

As for
$$
   \elam{n}{\sum_{i=0}^{4 \cdot n}\, \frac{i}{4}}
   :
   \mathbb{N} \rightarrow \mathbb{N},
$$

\noindent
it is enumerated by the stream of hexagonal numbers~\cite{OEIS:A000384},
\ie, $\estreamsix{0}{1}{6}{15}{28}{45}$.
So it is equivalent to $\elam{n}{2 \cdot n^2 - n}$.
In words, summing four times as many successive quarters yields an hexagonal number:

\begin{theorem}[Hexagonal numbers as bounded summations of integer quotients] \ \\
$\forall n \in \mathbb{N},
 \sum_{i=0}^{4 \cdot n} \, \equofloor{i}{4}
 =
 2 \cdot n^2 - n$
\end{theorem}

And since
$$
   \elam{n}{\sum_{i=0}^{1 \cdot n}\, \frac{i}{1}}
   =
   \elam{n}{\sum_{i=0}^n\, i}
   :
   \mathbb{N} \rightarrow \mathbb{N}
$$

\noindent
is enumerated by the stream of triangular numbers~\cite{OEIS:A000270},
\ie, $\estreamsix{0}{1}{3}{6}{10}{15}$,
it is equivalent to $\elam{n}{\frac{n \cdot (n + 1)}{2}}$.
In words, summing the first successive natural numbers yields a
triangular number:

\begin{theorem}[Triangular numbers as bounded summations of natural numbers] \ \\
$\forall n : \mathbb{N},
 \sum_{i=0}^{1 \cdot n} \equofloor{i}{1}
 =
 \frac{n \cdot (n + 1)}{2}$
\end{theorem}

These observations lead one to the following theorem:

\begin{theorem}[Polygonal numbers as bounded summations of increasing quotients (June 2023)]
For any given positive integer $k$ and for all natural numbers $n$,
$\sum_{i=0}^{k \cdot (n + 1)}\, \equofloor{i}{k}$
is a $k + 2$nd polygonal number.
Put otherwise,
$\sum_{i=0}^{k \cdot (n + 1)}\, \equofloor{i}{k}
 =
 \sum_{i=0}^n (k \cdot i + 1)
=
k \cdot \sum_{i=0}^n i + n + 1
$
in traditional summative form and in terms of triangular numbers.
\end{theorem}

\noindent
To prove the first equality, the following
pair of
identities comes handy:
$\forall x : \mathbb{N},
 \forall y : \mathbb{N},
 \forall f : \mathbb{N} \rightarrow \mathbb{N} \rightarrow \mathbb{N}$,
$$
\sum_{i=0}^{(x + 1) \cdot (y + 1)}\eapp{f}{i}
=
\sum_{i=0}^{x \cdot (y + 1) + y}\eapp{f}{i} \:\;+\;\, \eappq{f}{(x + 1) \cdot (y + 1)}
\;\;\;\;\;\wedge\;\;\;\;\;
\sum_{i=0}^{x \cdot (y + 1) + y}\eapp{f}{i}
=
\sum_{i=0}^x\,\sum_{j=0}^y\,\eappq{f}{i \cdot (y + 1) + j}
$$

\noindent
(To start the proof, observe that since $k$ is a positive natural number,
there exists a natural number $k'$ such that $k = k' + 1$,
and then consider $\sum_{i=0}^{(n + 1) \cdot (k' + 1)}\, \equofloor{i}{k' + 1}$.
Extra hint: $\sum_{j=0}^{k'}\,\equofloor{i \cdot (k' + 1) + j}{k' + 1} = \sum_{j=0}^{k'}\,i$.)

In words, polygonal numbers can be characterized as bounded summations of increasing quotients of the polygonal order, minus 2.
In each sum, $k$ occurs both in the upper bound of the summation as a multiplier and in its body as a divisor.

Now what does this characterization buy us?
In the author's admittedly limited experience, not terribly much:
It is simple to formalize in Coq and it is well suited to reason, \eg, about the difference between two polygonal numbers,
but otherwise it does not appear to provide an insightful new edge to prove other properties,
nor an ingenious new edge to compute polygonal numbers more efficiently.


\subsection{Conclusion}

Characterizing polygonal numbers as sums of increasing quotients (twice
as many successive halves, three times as many successive thirds, etc.)
is aligned with the additive concerns of the Greeks, thought-provokingly
simple, and pleasingly uniform -- perhaps because its nature is
computational.
Still, this characterization does not propel number theory forward
nor does it provide a more efficient way of computing polygonal numbers.
At any rate, for all its simplicity, this characterization is new,
which is unexpected considering that polygonal numbers date back to Pythagoras.
As such, it makes a fun Curiosa in memory of Alfred Moessner.

\end{document}